\documentclass{article}

    \PassOptionsToPackage{numbers, compress}{natbib}
\usepackage[final]{neurips_2019}

\usepackage[utf8]{inputenc} % allow utf-8 input
\usepackage[T1]{fontenc}    % use 8-bit T1 fonts
\usepackage{hyperref}       % hyperlinks
\usepackage{url}            % simple URL typesetting
\usepackage{booktabs}       % professional-quality tables
\usepackage{amsfonts}       % blackboard math symbols
\usepackage{nicefrac}       % compact symbols for 1/2, etc.
\usepackage{microtype}      % microtypography
\usepackage{amsmath,amsthm,graphicx,enumitem,natbib}
\usepackage{caption}
\usepackage{subcaption}

\usepackage{ifthen}
\newboolean{full}
\setboolean{full}{true}% true for full version, false for CRV version without appendix

\numberwithin{equation}{section}
\newtheorem{theorem}{Theorem}[section]
\newtheorem{lemma}{Lemma}[section]

\newtheorem{proposition}{Proposition}[section]
\theoremstyle{remark}
\newtheorem{remark}{Remark}[section]
\title{Nonzero-sum Adversarial Hypothesis Testing Games}%\ifthenelse{\boolean{full}}{}{\footnote{The full version of this paper is available at the following link: \url{XXX}}}

 \author{
  Sarath Yasodharan %\thanks{Use footnote for providing further information about author (webpage, alternative address)---\emph{not} for acknowledging funding agencies.} 
  \\
   Department of Electrical Communication Engineering\\
   Indian Institute of Science\\
   Bangalore 560~012, India  \\
   \texttt{sarath@iisc.ac.in} \\
  \And
   Patrick Loiseau % \thanks{test}
   \\
   Univ. Grenoble Alpes, Inria, CNRS, Grenoble INP, LIG \& MPI-SWS \\
   700 avenue Centrale\\
    Domaine Universitaire\\
    38400 St Martin d’H\'eres, France
    \\
   \texttt{patrick.loiseau@inria.fr} \\
}

\begin{document}

\maketitle

\begin{abstract}
We study nonzero-sum hypothesis testing games that arise in the context of adversarial classification, in both the Bayesian as well as the Neyman-Pearson frameworks. We first show that these games admit mixed strategy Nash equilibria, and then we examine some interesting concentration phenomena of these equilibria. Our main results are on the exponential rates of convergence of classification errors at equilibrium, which are analogous to the well-known Chernoff-Stein lemma and Chernoff information that describe the error exponents in the classical binary hypothesis testing problem, but with parameters derived from the adversarial model. The results are validated through numerical experiments.
\end{abstract}

\section{Introduction}
\label{section:introduction}

Classification is a simple but important task that has numerous applications in a variety of domains such as computer vision or security. A traditional assumption that is used in the design of classification algorithms is that the input data is generated without knowledge of the classifier being used, hence the data distribution is independent of the classification algorithm. This assumption is no longer valid in the presence of an adversary, as an adversarial agent can learn the classifier and deliberately alter the data such that the classifier makes an error. This is the case in particular in security applications where the classifier's goal is to detect the presence of an adversary from the data  it observes.

%There are two main game-theoretic settings in which this problem has been studied in the past. 
%In both these frameworks, the goal of the defender is to detect the presence of the attacker. 
Adversarial classification has been studied in two main settings. 
The first focuses on adversarial versions of a standard classification task in machine learning, where the adversary attacks the classifier (defender/decision maker) by directly choosing vectors from a given set of data vectors; whereas the second focuses on adversarial hypothesis testing, where the adversary (attacker) gets to choose a distribution from a set of distributions and independent data samples are generated from this distribution. The main differences of the latter framework from the former are that: (i) the adversary only gets to choose a distribution (rather than the actual attack vector) and data is generated independently from this distribution, and (ii) the defender makes a decision only once after it observes a whole data sequence instead of making a decision for each individual data sample it receives. Both of these frameworks have applications in a variety of domains, but prior literature has mainly focused on the first setting; see Section~\ref{section:relatedwork} for a description of the related literature.

In this paper, we focus on the setting of adversarial hypothesis testing. To model the interaction between the attacker and defender, we formulate a nonzero-sum two-player game between the adversary and the classifier where the adversary picks a distribution from a given set of distributions, and data is generated independently from that distribution (a non-attacker always generates data from a fixed distribution). The defender on his side makes a decision based on observation of $n$ data points. Our model can also be viewed as a game-theoretic extension of the classical binary hypothesis testing problem where the distribution under the alternate hypothesis is chosen by an adversary. Based on our game model, we are then able to extend to the adversarial setting the main results of the classical hypothesis testing problem (see Section~\ref{section:classical-hypothesis-testing}) on the form of the best decision rule and on the rates of decrease of classification errors. More specifically, our contributions can be summarized as follows: \vspace{-1mm}%\\[-6mm]
\begin{enumerate}
\item We propose nonzero-sum games to model adversarial hypothesis testing problems in a flexible manner.\vspace{-1mm}%\\[-4mm]

\item We show existence of mixed strategy Nash equilibria in which the defender employs certain likelihood ratio tests similar to that used in the classical binary hypothesis testing problem.\vspace{-1mm}%\\[-4mm]

\item We show that the classification errors under all Nash equilibria for our hypothesis testing games decay exponentially fast in the number of data samples. We analytically obtain these error exponents, and it turns out that they are same as those arising in certain classical hypothesis testing problem, with parameters derived from the adversarial model.\vspace{-1mm}%\\[-4mm]

\item We illustrate the results, in particular the importance of some assumptions, using simulations.\vspace{-2mm}%\\[-4mm]

\end{enumerate}
Throughout our analysis, an important difficulty lies in that the strategy spaces of both the players are uncountable; we believe, however, that it is an important feature of the model to be realistic.

\subsection{Related Work}
\label{section:relatedwork}

Adversarial classification and the security of machine learning have been studied extensively in the past decade, see e.g., \cite{dalvi-etal-04,Lowd05a,Barreno10a,Huang11a,Li15a,Papernot18a}; here we focus only on game-theoretic approaches to tackle the problem. Note that, besides the adversarial learning problem, game theory has been successfully used to tackle several other security problems such as allocation of monitoring resources to protect targets, see e.g., \cite{chen-leneutre-14,korzhyk-etal-11}. We review here only papers relating to classification. 

%Another setting that is closely related to the adversarial classification setting is that of security games. Here, there is a set of possible targets that the attacker can attack, and the defender protects these targets by using limited resources at her disposal. The utilities capture how many targets have been attacked, and how many targets have been saved from attack by the defender. In~\cite{chen-leneutre-14}, the authors analyze Nash equilibria of such games. Korzhyk et al.~\cite{korzhyk-etal-11} examines some properties of Nash equilibria, and connections with strong Stackelberg equilibrium for security games. 

A number of game-theoretic models have appeared in the past decade to study the adversarial classification problem in the classical setting of classification tasks in machine learning. \cite{dalvi-etal-04} studies the best response in an adversarial classification game, where the adversary is allowed to alter training data. A number of zero-sum game models were also proposed where the attacker is restricted on the amount of modifications he can do to the training set, see~\cite{Kantarcioglu11a,Zhou12a,Zhou14a}. \cite{bruckner-etal-12} studies the problem of choosing the best linear classifier in the presence of an adversary (a similar model is also studied in \cite{Bruckner11a}) using a nonzero-sum game, and shows the existence of a unique pure strategy Nash equilibrium. Similar to our formulation, the strategy sets in this case are uncountable, and therefore showing the existence and uniqueness of Nash equilibrium needs some work. However, in our formulation, there may not always exist a Nash equilibrium in pure strategies, which makes the subsequent analysis of error exponents  more difficult. \cite{lisy-etal-14} studies an adversarial classification game where the utilities of the players are defined by using ROC curves. The authors study Nash equilibria for their model and provide numerical discretization techniques to compute the equilibria.
\cite{dritsoula-etal-17} studies a nonzero-sum adversarial classification game where the defender has no restriction on the classifier, but the attacker is limited to a finite set of vectors. The authors show that the defender can, at equilibrium, use only a small subset of ``threshold classifiers'' and characterize the equilibrium through linear programming techniques. 
In our model, the utility functions share similarities with that of~\cite{dritsoula-etal-17}, but we work in the hypothesis testing framework and with uncountable action sets, which completely modifies the analysis. 
Several studies appeared recently on ``strategic classification'', where the objective of the attacker(s) is to improve the classification outcome in his own direction, see~\cite{Hardt16a,Dong18a}. 

%Dritsoula et al.~\cite{dritsoula-etal-17} study a non-zero-sum game that arise in adversarial classification. In this setting, a legitimate user generates data from a certain distribution, whereas an adversary can attack the classifier using attack vectors, and the adversary gains certain revenue based on the attack vector that she uses. It is shown that, it suffices to focus on threshold classifiers, i.e., those that declare the presence of an adversary when the revenue associated with the attack vector is larger than a threshold. The authors also provide the structure of all mixed strategy Nash equilibria of this game by studying the associated linear programs. In our model, the utility functions of the players are motivated by that in~\cite{dritsoula-etal-17}. However, the additional randomness from the adversary makes our game much more difficult to analyze, since the strategy sets become uncountable, and the utility functions lose some interesting monotonicity properties observed in~\cite{dritsoula-etal-17}, that helped the authors to focus on threshold classifiers and study the structure of mixed strategy equilibria.

On the other hand, adversarial hypothesis testing has been studied by far fewer authors. \cite{barni-tondi-13} studies a source identification game in the presence of an adversary, where the classifier needs to distinguish between two source distributions $P_0$ and $P_1$ in which the adversary can corrupt samples from $P_0$ before it reaches the classifier. They show that the game has an asymptotic Nash equilibrium when the number of samples becomes large, and compute the error exponent associated with the false negative probability. \cite{barni-tondi-14} and \cite{tondi-etal-18} study further extensions of this framework.

A (non game-theoretic) hypothesis testing problem in an adversarial setting has been studied by~\cite{brandao-etal-14}, which is the closest to our work. Here, there are two sets of probability distributions and nature outputs a fixed number of independent samples generated by using distributions from either one of these two sets. The goal of the classifier is to detect the true state of nature. The authors derive error exponents associated with the classification error, in both Bayesian and Neyman-Pearson frameworks using a worst-case maxmin analysis. Although we restrict to i.i.d. samples and let the non-attacker play a single distribution, we believe that our nonzero-sum game model with flexible utilities can better capture the interaction between adversary and classifier. 
There also exists extensive prior work within the statistics literature \cite{ingster-suslina} on minimax hypothesis testing, which relates to our paper, but we defer a discussion of how our work differs from it to after we have exposed the details of our model.

Game-theoretic models were also used to study adversarial classification in a sequential setting, see~\cite{soper-musacchio-14,bao-etal-11,lye-wing-05}, but with very different techniques and results. % are very different from the present paper. 

%Soper and Musacchio~\cite{soper-musacchio-14} model the presence of an adversary as a Brownian motion with drift, where the drift parameter can be controlled by the adversary, and the utility of the adversary depends on the drift parameter that she chooses. The goal of the classifier is to detect the presence of an adversary as quickly as possible. Their main result is on existence of a pure strategy Nash equilibrium for this problem. Another adversarial classification problem in the sequential case has been studied numerically by Bao et al.~\cite{bao-etal-11} where the classifier wants to distinguish between the presence of a spammer and a spy. A numerical study of a stochastic game that arise in the context of network security has been studied by Lye and Wing~\cite{lye-wing-05}. The authors numerically compute the best response of the players as well as Nash equilibrium of their game.

\section{Basic Setup and Hypothesis Testing Background} 
\label{section:classical-hypothesis-testing}

In this section, we present the basic setup results in classical binary hypothesis testing. 

Throughout the paper, we consider an alphabet set $\mathcal{X}$ that we assume finite. In a classical hypothesis testing problem, we are given two distribution $p$ and $q$, and a realization of a sequence of independent and identically distributed random variables $X_1, \dots, X_n$, which are distributed as either $p$ (under hypothesis $H_0$) or $q$ (under hypothesis $H_1$). Our goal is to distinguish between the two alternatives:
\begin{align*}
H_0:  X_1, X_2, \dots, X_n \text{ i.i.d.}\sim p \quad
\text{versus}
\quad
H_1:  X_1, X_2, \dots, X_n  \text{ i.i.d.}\sim q.
\end{align*}
In this setting, we could make two possible types of errors: (i) we declare $H_1$, whereas the true state of nature is $H_0$ (Type I error, or false alarm), and (ii) we declare $H_0$ whereas the true state of nature is $H_1$ (Type II error, or missed detection). Note that one can make one of these errors arbitrarily small at the expense of the other by always declaring $H_0$ or $H_1$. 

The trade-off between the two types of errors can be captured  using two frameworks. If we have knowledge on the prior probabilities of the two hypotheses, then we can seek a decision rule that minimizes the average probability of error (this is the Bayesian framework). On the other hand, if we do not have any information on the prior probabilities, then we can fix  $\varepsilon > 0$ and seek a decision rule that minimizes the Type II error among all decision rules whose Type I error is at most $\varepsilon$ (this is the Neyman-Pearson framework). In both of these frameworks, it can be shown that the optimal test is a likelihood ratio test, i.e., given $\mathbf{x}^n = (x_1, \dots, x_n)$ we compute the likelihood ratio $\frac{q(\mathbf{x}^n)}{p(\mathbf{x}^n)}$ and compare it to a threshold to  make a decision (with possible randomization at the boundary in the Neyman-Pearson framework). Here, $p(\mathbf{x}^n)$ (resp.~$q(\mathbf{x}^n$)) denotes the probability of observing the $n$-length word $\mathbf{x}^n$ under the distribution $p$ (resp.~$q$). See Section II.B and II.D in~\cite{poor} for an introduction to hypothesis testing.

For large enough $n$, by the law of large numbers, the fraction of $i$ in an observation $\mathbf{x}^n$ is very close to $p(i)$ (resp. $q(i)$) under $H_0$ (resp.~under $H_1$), for each $i \in \mathcal{X}$. Therefore, one anticipates that the probability of correct decision is very close to $1$ for large enough $n$. Hence, one can study the rate at which the errors go to $0$ as $n$ becomes large. It is shown that, under both frameworks, the error decays exponentially in $n$. In the Bayesian framework, the error exponent associated with the average probability of error is $-\Lambda^*_0(0)$, where $\Lambda^*_0(\cdot)$ is the Fenchel-Legendre transform of the log-moment generating function of the random variable $\frac{q(X)}{p(X)}$ under $H_0$, i.e., when $X\sim p$. In the Neyman-Pearson case, the error exponent associated with the Type II error is $-D(p||q)$ where $D$ is the relative entropy functional. The above error exponents are known as Chernoff information and Chernoff-Stein lemma, respectively (see Section~3.4 in~\cite{dembo-zeitouni} for the analysis on error exponents).

In this work, we propose extensions of the classical hypothesis testing framework to an adversarial scenario modeled as a game, both in the Bayesian and in the Neyman-Pearson frameworks; and we investigate how the corresponding results are modified. Due to space constraints, we present only the model and results for the Bayesian framework in the main body of the paper. The corresponding analysis for the Neyman-Pearson framework follows similar ideas and is relegated to \ifthenelse{\boolean{full}}{Appendix~\ref{section:neyman-pearson}}{Appendix~A of the full version of this paper \cite{Sarath19aFull}}. The proofs of all results presented in the paper (and in \ifthenelse{\boolean{full}}{Appendix~\ref{section:neyman-pearson}}{Appendix~A of the the full version \cite{Sarath19aFull}}) can be found in \ifthenelse{\boolean{full}}{Appendix~\ref{section:proofs}}{Appendix~B of the full version \cite{Sarath19aFull}}.

\section{Hypothesis Testing Game in the Bayesian Framework}
\label{section:bayesian}

In this section, we formulate a one-shot adversarial hypothesis testing game in the Bayesian framework, motivated by security problems where there might be an attacker who modifies the data distribution and a defender who tries to detect the presence of the attacker. Game theoretic modelling of such problems has found great success in understanding the behavior of the agents via equilibrium analysis in many applications, see Section~\ref{section:relatedwork}. We first present the model and then elaborate on its motivations and on how it relates to related works in statistics. 

%Our game is motivated by problems from security where there is an attacker who tries to access the system and a defender who tries to detect the presence of the attacker. We use game theoretic modelling to understand these situations as both the agents are rational and the objective of an agent depends on the action of the other. Game theoretic modelling of such problems has found great success in understanding the behavior of the agents via equilibrium analysis; some areas of application include multimedia forensics (which is well modeled by adversarial hypothesis testing) (\cite{barni-tondi-13, barni-tondi-14, tondi-etal-18}), intrusion detection (\cite{soper-musacchio-14}),   adversarial classification~(\cite{dritsoula-etal-17}), adversarial learning~(\cite{bruckner-etal-12, Huang11a, Lowd05a, Zhou14a, Zhou12a}), network security~(\cite{chen-leneutre-14}).

\subsection{Problem Formulation}
Let $\mathcal{X} = \{0,1, \dots, d-1\}$ denote the alphabet set with cardinality $d$, and  let $M_1(\mathcal{X})$  denote the space of probability distributions on $\mathcal{X}$. Fix $n \geq 1$. 

The game is played as follows. There are two players: the external agent and the defender. The external agent can either be a non-attacker or an attacker. In the Bayesian framework, we assume that the external agent is an attacker with probability $\theta$, and a non-attacker (normal user) with probability $1-\theta$. The non-attacker is not strategic and she does not have any adversarial objective. If the external agent is a non-attacker, she generates $n$ samples independently from the distribution $p$. If the external agent is an attacker, she picks a distribution $q$ from a set of distributions $Q \subseteq M_1(\mathcal{X})$ and generates $n$ samples independently from $q$. The defender, upon observing the $n$-length word generated by the external agent, wants to detect the presence of the attacker.

Throughout the paper, a decision rule implemented by the defender is denoted by $\varphi : \mathcal{X}^n \to  [0,1]$, with the interpretation that $\varphi(\mathbf{x}^n)$ is the probability with which hypothesis $H_1$ is accepted (i.e., the presence of an adversary is declared) when the defender observes the $n$-length word $\mathbf{x}^n = (x_1, \dots, x_n)$. We say that a decision rule $\varphi$ is deterministic if $\varphi (\mathbf{x}^n)\in \{0, 1\}$ for all $\mathbf{x}^n \in \mathcal{X}^n$. 

To define the game, let the attacker's strategy set be $Q \subseteq M_1(\mathcal{X})$, and that of the defender be
\begin{equation*}
\Phi_n= \{\varphi:\mathcal{X}^n \to [0,1]\},
\end{equation*}
which is the set of all randomized decision rules on $n$-length words.

To define the utilities, consider the attacker first. We assume that there is a cost associated with choosing a distribution from $Q$ which we model using a cost function $c : Q \to \mathbb{R}_+$. The goal of the attacker is to fool the defender as much as possible, i.e., he wants to maximize the probability that the defender classifies an $n$-length word as coming from the non-attacker whereas it is actually being generated by the attacker. To capture this, the utility of the attacker when she plays the pure strategy $q \in Q$ and the defender plays the pure strategy $\varphi \in \Phi_n$ is defined as
\begin{align}
u^A_n(q,\varphi) = \sum_{\mathbf{x}^n} (1-\varphi(\mathbf{x}^n)) q(\mathbf{x}^n) - c(q),
\label{eqn:attacker_revenue_pure}
\end{align}
where $q(\mathbf{x}^n)$ denotes the probability of observing the $n$-length word $\mathbf{x}^n$ when the symbols are generated independently from the distribution $q$. 

For the defender, the goal is to minimize the classification error. Similar to the classical hypothesis testing problem, there could be two types of errors: (i) the external agent is actually a non-attacker whereas the defender declares that there is an attack (Type I error, or false alarm), and (ii) the external agent is an attacker whereas the defender declares that there is no attack (Type II error, or missed detection). The goal of the defender is to minimize a weighted sum of the above two types of errors. After suitable normalization, we define the utility of the defender as
\begin{align}
u^D_n(q,  \varphi) =  -\left(\sum_{\mathbf{x}^n} (1-\varphi(\mathbf{x}^n)) q(\mathbf{x}^n ) + \gamma \sum_{\mathbf{x}^n}\varphi(\mathbf{x}^n)p(\mathbf{x}^n) \right),
\label{eqn:defender_revenue_pure}
\end{align}
where $\gamma >0$ is a constant that captures the exogenous probability of attack (i.e., $\theta$), as well as the relative weights given to the error terms.

We denote our Bayesian hypothesis testing game with utility functions (\ref{eqn:attacker_revenue_pure}) and (\ref{eqn:defender_revenue_pure}) by $\mathcal{G}^B(d,n)$.
With a slight abuse of notation, we denote by $u^A_n(\sigma^A_n, \sigma^D_n)$ and $u^D_n(\sigma^A_n, \sigma^D_n)$, the utility of the players under a mixed strategy $(\sigma^A_n, \sigma^D_n)$, where $\sigma^A_n \in M_1(Q)$, and $\sigma^D_n \in M_1(\Phi_n)$.

For our analysis of game $\mathcal{G}^B(d,n)$, we will make use of the following assumptions:\vspace{-1mm}
\begin{enumerate}[label=({A\arabic*})]
\item $Q$ is a closed subset of $M_1(\mathcal{X})$, and $p \notin Q$.  \label{assm:a1}
\item $p(i) > 0$ for all $i \in \mathcal{X}$. Furthermore, for each $q \in Q$, $q(i) > 0$ for all $i \in \mathcal{X}$. \label{assm:a1b}
\item \label{assm:a2} $c$ is continuous on $Q$, and there exists a unique $q^* \in Q$ such that\vspace{-1mm}
\begin{equation*}
q^* = \arg\min_{q \in Q} c(q).
\end{equation*}\vspace{-3mm}
\item The point $p$ is distant from the set $Q$ relative to the point $q^*$, i.e.,
\begin{align*}%
\{\mu \in M_1(\mathcal{X}):D(\mu||p) \leq  D(\mu ||q^*)\} \cap Q = \emptyset,
\end{align*}\label{assm:a3}%
where $D(\mu || \nu) = \sum_{i \in \mathcal{X}} \mu(i) \log \frac{\mu(i)}{\nu(i)}, \mu, \nu \in M_1(\mathcal{X})$, denotes the relative entropy between the distributions $\mu $ and $\nu$.
\end{enumerate}
Note that~\ref{assm:a1} and \ref{assm:a1b} are very natural. In~\ref{assm:a1b}, if $p(i) = 0$ for some $i \in \mathcal{X}$ and $q(i) > 0$ for some $q \in Q$, then the adversary will never pick $q$, as the defender can easily detect the presence of the attacker by looking for element $i$. On the other hand, if $p(i) = 0$ and $q(i) = 0$ for all $q \in Q$, we may consider a new alphabet set without $i$. In~\ref{assm:a2}, continuity of the cost function $c$ is natural and we do not assume any extra condition other than the requirement that there is a unique minimizer.
Assumption~\ref{assm:a3} is used to show certain property of the equilibrium of the defender, which is later used in the study of error exponents associated with classification error. Specifically, Assumption~\ref{assm:a3} is used in the proofs of Lemma~\ref{lemma:errorgoesto0}, Lemma~\ref{lemma:eqconcentration} and  Theorem~\ref{theorem:errorexponent};  all other results are valid without this assumption. We will further discuss the role of \ref{assm:a2} and \ref{assm:a3} in Section~\ref{section:error_exp} after Theorem~\ref{theorem:errorexponent}.

\subsection{Model discussion}

%application + modeling as cost in practice
Our setting is that of adversarial hypothesis testing, where the attacker chooses a distribution and points are then generated i.i.d. according to it. This is a reasonable model in applications such as multimedia forensics (where one tries to determine if an attacker has tampered with an image from signals that can be modeled as random variables following an image-dependent distribution) or biometrics (where again one tries to detect from random signals whether the perceived signals do come from the characteristics of a given individual or they come from tampered characteristics)---see more details about these applications in \cite{barni-tondi-13,barni-tondi-14,tondi-etal-18}. In such applications, it is reasonable that different ways of tampering have different costs for the attacker and that one can estimate those costs for a given application at least to some extent. Modeling the attacker's utility via a cost function is classical in other settings, for instance in adversarial classification \cite{dritsoula-etal-17,soper-musacchio-14,bruckner-etal-12} and experiments with real-world applications where a reasonable cost function can be estimated has been done, for instance, in~\cite{bruckner-etal-12}.

%We remark that, modeling the utility of the attacker via a cost function, notably in the adversarial classification literature where the attacker chooses data points, has been done in the past in various papers, see e.g., \cite{dritsoula-etal-17,soper-musacchio-14,bruckner-etal-12} (also, experiments with real-world applications where a reasonable cost function can be estimated has been done in~\cite{bruckner-etal-12}). Our setting is that of adversarial hypothesis testing, where the attacker chooses a distribution (and points are then generated i.i.d. according to it). This is a reasonable model in several applications such as multimedia forensics (where one tries to determine if an attacker has tempered with an image from signals that can be modeled as random variables following an image-dependent distribution) or biometrics (where again one tries to detect from random signals whether the perceived signal do come from the characteristics of a given individual or they come from tempered characteristics)- see more details about these applications in \cite{barni-tondi-13,barni-tondi-14,tondi-etal-18}. In such applications, it is also reasonable that different temperings have different costs for the attacker and that one can estimate those costs for a given application at least to some extent.

%difference to composite testing

Our setting is very similar to that of a composite hypothesis testing framework where nature picks a distribution from a given set and generates independent samples from it. However, in such problems, one does not model a utility function for the nature/statistician and one is often interested in existence and properties of uniformly most powerful test or locally most powerful test (depending on the Bayesian or frequentist approach; see Section~II.E~in~\cite{poor}). In contrast, here, we specifically model the utility functions for the agents and investigate the behavior at Nash equilibrium using very different analysis, which is more natural in adversarial settings where two rational agents interact.  

%difference to minimax
Our setting also coincides with the well-studied setting of minimax testing~\cite{ingster-suslina} when $c(q)=0$ for all $q \in Q$ (and hence every $q$ is a minimizer of $c$). Note, however, that this case is not included in our model due to Assumption~\ref{assm:a2}---rather we study the opposite extreme where $c$ has a unique minimizer. Our results are not an easy extension of the classical results because our game is now a nonzero-sum game (whereas the minimax setting corresponds to a zero-sum game). We can therefore not inherit any of the nice properties of zero-sum games; in particular we cannot compute the NE and we instead have to prove properties of the NE (e.g., concentration) without being able to explicitly compute it. In fact, our results too are quite different since we show that the error rate is the same as a simple test where $H_1$ would contain only $q^*$, which is different from the classical minimax case.

%When $c(q)=0$ for all $q \in Q$, our setting coincides with the standard minimax setting, see e.g.,~\cite{ingster-suslina}. Our paper, with a non-constant function c, indeed represents an extension of that setting. However, note that our results are not easy extension of the classical results because our game is now a nonzero-sum game (whereas the minimax setting corresponds to a zero-sum game). We can therefore not inherit any of the nice properties of zero-sum games; in particular we are not able to compute the NE and we have instead to prove properties of the NE (e.g., concentration) without being able to explicitly compute it. In fact, our results too are quite different since we show that the error rate is the same as a simple test where $H_1$ would contain only $q^*$---which is not the same as for the classical minimax case.

%not sequential here

Finally, in our model we fix the sample size $n$, i.e., the defender makes a decision only after observing all $n$ samples. We restrict to this simpler setting since it has applications in various domains (see~Section~\ref{section:relatedwork}), and understanding the equilibrium of such games leads to interesting and non-trivial results. We leave the study of a sequential model where the defender has the flexibility to choose the number of samples for decision making as future work.

%Another natural setting is where the defender has the flexibility to choose the number of samples for decision making. Such a sequential model will be studied in a future paper.

\section{Main Results}

%In this section, we present our main results for the hypothesis testing game in the Bayesian framework. 

\subsection{Mixed Strategy Nash Equilibrium for $\mathcal{G}^B(d,n)$}

We first examine the Nash equilibrium for $\mathcal{G}^B(d,n)$. Note that the strategy sets of both the attacker and the defender are uncountable, hence it is a priori not clear whether our game has a Nash equilibrium.

Towards this, we equip the set $\Phi_n$ of all randomized decision rules with the sup-norm metric, i.e.,
\begin{align*}
d_n(\varphi_1,\varphi_2) = \max_{\mathbf{x}^n \in \mathcal{X}^n} |\varphi_1(\mathbf{x}^n) -  \varphi_2(\mathbf{x}^n)|,
\end{align*}
for $\varphi_1, \varphi_2 \in \Phi_n$. It is easy to see that the set $\Phi_n$ endowed with the above metric is a compact metric space.
We also equip $M_1(\mathcal{X})$ with the usual Euclidean topology on $\mathbb{R}^d$, and equip $Q$ with the subspace topology. Also, for studying the mixed extension of the game, we equip the spaces $M_1(Q)$ and $M_1(\Phi_n)$ with their corresponding weak topologies. Product spaces are always equipped with the corresponding product topology.

%Towards studying the existence of mixed strategy Nash equilibrium (NE), 
We begin with a simple continuity property of the utility functions. %, proof of which can be found in Appendix~1.
\begin{lemma}
Assume~\ref{assm:a1}-\ref{assm:a2}. Then, the utility functions $u^A_n$ and $u^D_n$ are continuous on $Q \times \Phi_n$.
\label{lemma:cts-payoffs}
\end{lemma}

We now show the main result of this subsection, namely existence and partial characterization of a NE for our hypothesis testing game.
\begin{proposition}
Assume~\ref{assm:a1}-\ref{assm:a2}. Then, there exists a mixed strategy Nash equilibrium for $\mathcal{G}^B(d,n)$. If $(\hat{\sigma}^A_n, \hat{\sigma}^D_n)$ is a NE, then so is $(\hat{\sigma}^A_n, \hat{\varphi}_n)$ where $\hat{\varphi}_n$ is the likelihood ratio test given by
\begin{align}
\hat{\varphi}_n(\mathbf{x}^n) = \left\{
\begin{array}{lll}
1, & \!\!\text{if } q_{\hat{\sigma}^A_n}(\mathbf{x}^n) - \gamma p(\mathbf{x}^n) & \!\!> 0, \\
\varphi_{\hat{\sigma}^D_n}, & \!\!\text{if } q_{\hat{\sigma}^A_n}(\mathbf{x}^n) - \gamma p(\mathbf{x}^n) & \!\!= 0, \\
0, & \!\!\text{if } q_{\hat{\sigma}^A_n}(\mathbf{x}^n) - \gamma p(\mathbf{x}^n) & \!\!< 0,
\end{array}
\right.
\end{align}\label{prop:defenderpure}%
where $q_{\hat{\sigma}^A_n}(\mathbf{x}^n) = \int q(\mathbf{x}^n) \hat{\sigma}^A_n(dq)$, and $ \varphi_{\hat{\sigma}^D_n} = \int \varphi(\mathbf{x}^n)\hat{\sigma}^{D}_n(d\varphi)$.
\end{proposition}
The existence of a NE follows from Glicksberg's fixed point theorem (see e.g.,~Corollary~2.4~in~\cite{reny-05}); for the form of the defender's equilibrium strategy, we have to examine the utility function $u^D_n$. %; see Appendix~1 for the details.
\begin{remark}
Note that we have considered randomization over $\Phi_n$ to show existence of a NE. Once this is established, we can then show the form of the strategy of the defender $\hat{\varphi}_n$ at equilibrium; the existence of a NE is not clear if we do not consider randomization over $\Phi_n$.
\end{remark}
\begin{remark}Note that the distribution $q_{\hat{\sigma}^A_n}$ on $\mathcal{X}^n$ cannot necessarily be written as an $n$-fold product distribution of some element from $M_1(\mathcal{X})$. Therefore, the test $\hat{\varphi}_n$ is slightly different from the usual likelihood ratio test that appears in the classical hypothesis testing problem where samples are generated independently.
\end{remark}
\begin{remark}
Apart from the conditions of the above proposition, a sufficient condition for existence of pure strategy Nash equilibrium is that the utilities are individually quasiconcave, i.e., $u^A_n(\cdot, \varphi)$ is quasiconcave for all $\varphi \in \Phi_n$, and $u^D_n(q,\cdot)$ is quasiconcave for all $q \in Q$. However, it is easy to check that the Type II error term is not quasiconcave in the attacker's strategy, and hence the utility of the attacker is not quasiconcave. Hence, a pure strategy Nash equilibrium is not guaranteed to exist---see numerical experiments in \ifthenelse{\boolean{full}}{Appendix~\ref{section:neyman-pearson}}{Appendix~C of the full version of this paper \cite{Sarath19aFull}}.
\end{remark}
\begin{remark}
Proposition~\ref{prop:defenderpure} does not provide any information about the structure of the attacker's strategy at a NE. We believe that obtaining the complete structure of a NE and computing it is a difficult problem in general because the strategy spaces of both players are uncountable (and there is no pure-strategy NE in general), and we cannot use the standard techniques for finite games. However, we emphasize that we are able to obtain error exponents at an equilibrium (see Theorem~\ref{theorem:errorexponent}) without explicitly computing the structure of a NE. Also, one could study Stackelberg equilibrium for our game $\mathcal{G}^B(d, n)$ to help solve computational issues, although we note that most of the security games literature using Stackelberg games assumes finite action spaces~(see, for example, \cite{korzhyk-etal-11}); however we do not address the study of Stackelberg equilibrium in this paper.
\end{remark}
\subsection{Concentration Properties of Equilibrium}
\label{section:eqconcentration}
We now study some concentration properties of the mixed strategy Nash equilibrium for the game $\mathcal{G}^B(d,n)$ for large $n$. The results in this section will be used later to show the exponential convergence of the classification error at equilibrium. %The proofs of all the lemmas in this section can be found in Appendix~1.

Let $e_n$ denote the classification error, i.e., $e_n(q, \varphi) = -u^D_n(q,\varphi), q \in Q, \varphi \in \Phi_n$. We begin with the following lemma, which asserts that the error at equilibrium is small for large enough $n$.
\begin{lemma}
Assume~\ref{assm:a1}-\ref{assm:a2}. Let $(\hat{\sigma}^A_n, \hat{\sigma}^D_n)_{n \geq 1}$ be a sequence such that, for each $n \geq 1$,  $(\hat{\sigma}^A_n, \hat{\sigma}^D_n)$ is a mixed strategy Nash equilibrium for $\mathcal{G}^B(d,n)$. Then,  $e_n(\hat{\sigma}^A_n, \hat{\sigma}^D_n) \to 0$ as  $n \to \infty$.
\label{lemma:errorbound}
\end{lemma}
The main idea in the proof is to let the defender play a decision rule whose acceptance set is a small neighborhood around the point $p$, and then bound $e_n(\hat{\sigma}^A_n, \hat{\sigma}^D_n)$ using the  error of the above strategy. %; see Appendix~1 for details.

We now show that the mixed strategy profile of the attacker $\hat{\sigma}^A_n$ converges weakly to the point mass at $q^*$ (denoted by $\delta_{q^*}$) as $n \to \infty$. This is a consequence of the fact that $q^*$ is the minimizer of $c$, and hence for large enough $n$, the attacker does not gain much by deviating from the point $q^*$.
\begin{lemma}Assume~\ref{assm:a1}-\ref{assm:a2}, and let $(\hat{\sigma}_n^A, \hat{\sigma}_n^D)_{n \geq 1}$ be as in Lemma~\ref{lemma:errorbound}. Then, $\hat{\sigma}^A_n \to \delta_{q^*}$ weakly as $n \to \infty$.  \label{lemma:weakconv}
\end{lemma}
Note that it is not clear from the above lemma that the equilibrium strategy of the attacker $\hat{\sigma}^A_n$ is supported on a small neighborhood around $q^*$ for large enough $n$. By playing a strategy $q$ that is far from $q^*$ we could still have $u_n^A(q, \hat{\sigma}^D_n) = u^A_n(\hat{\sigma}^A_n, \hat{\sigma}^D_n)$, since the error term in $u^A_n$ could compensate for the possible loss of utility from the cost term. We now proceed to show that this cannot happen under Assumption~\ref{assm:a3}. We first argue that the equilibrium error is small even when the attacker deviates from her equilibrium strategy.
\begin{lemma}\label{lemma:errorgoesto0}%
Assume~\ref{assm:a1}-\ref{assm:a3}, and 
let $(\hat{\sigma}^A_n, \hat{\sigma}^D_n)_{n \geq 1}$ be as in Lemma~\ref{lemma:errorbound}. Then,
\begin{align*}
\sup_{q \in Q} e_n(q, \hat{\sigma}^D_n) \to 0 \text{ as } n \to \infty.
\end{align*}\vspace{-3mm}
\end{lemma}
We are now ready to show the concentration of the attacker's equilibrium:
\begin{lemma}
Assume~\ref{assm:a1}-\ref{assm:a3}, and let $(\hat{\sigma}^A_n, \hat{\sigma}^D_n)_{n \geq 1}$ be as in Lemma~\ref{lemma:errorbound}. Let $(q_n)_{n \geq 1}$ be a sequence such that $q_n \in \text{supp}(\hat{\sigma}^A_n)$ for each $n \geq 1$. Then, $q_n \to q^*$ as $n \to \infty$. 
\label{lemma:eqconcentration}
\end{lemma}
The above concentration phenomenon is a consequence of the uniqueness of $q^*$ and Assumption~\ref{assm:a3}. The main idea in the proofs of Lemma~\ref{lemma:errorgoesto0} and Lemma~\ref{lemma:eqconcentration} is to essentially show that, for large enough $n$, the acceptance region of $H_0$ under (any) mixed strategy Nash equilibrium does not intersect the set $Q$. If we do not assume~\ref{assm:a3}, then the decision region at equilibrium could intersect $Q$, and we may not have the concentration property in the above lemma (we will still have the convergence property in Lemma~\ref{lemma:weakconv} though, which does not use \ref{assm:a3}).

\subsection{Error Exponents}
\label{section:error_exp}
With the results on concentration properties of the equilibrium from the previous section, we are now ready to examine the error exponent associated with the classification error at equilibrium. Let $\Lambda_0$ denote the log-moment generating function of the random variable $\frac{q^*(X)}{p(X)}$ under $H_0$, i.e., when $X\sim p$: $\Lambda_0(\lambda) \!=\! \log \sum_{i \in \mathcal{X}} \exp{\left(\lambda \frac{q^*(i)}{p(i)}\right)} p(i),\, \lambda \in \mathbb{R}$. Define its Fenchel-Legendre transform
\begin{align*}
\Lambda^*_0(x) = \sup_{\lambda \in \mathbb{R}} \{\lambda x - \Lambda_0(\lambda)\}, \ x \in \mathbb{R}. 
\end{align*}
Our main result in the paper (for the Bayesian case) is the following theorem.
\begin{theorem}\label{theorem:errorexponent}%
Assume~\ref{assm:a1}-\ref{assm:a3}, and let $(\hat{\sigma}^A_n, \hat{\sigma}^D_n)_{n \geq 1}$ be as in Lemma~\ref{lemma:errorbound}. Then,
\begin{align*}
\lim_{n \to \infty} \frac{1}{n} \log e_n(\hat{\sigma}^A_n, \hat{\sigma}^D_n) = -\Lambda_0^*(0).
\end{align*}
\end{theorem}%
Our approach to show this result is via obtaining asymptotic lower and upper bounds for the classification error at equilibrium $e_n(\hat{\sigma}^A_n, \hat{\sigma}^D_n)$. Since we do not have much information about the structure of the equilibrium, we first let one of the players deviate from their equilibrium strategy, so that we can estimate the error corresponding to the new pair of strategies, and then use these estimates to compute the error rate at equilibrium. The lower bound easily follows by letting the attacker play the strategy $q^*$, and using the error exponent in the classical hypothesis testing problem between $p$ and $q^*$. For the upper bound, we let the defender play a specific deterministic decision rule, and make use of the concentration properties of the equilibrium of the attacker in Section~\ref{section:eqconcentration}. %For details, see Appendix~1.

Thus, we see that the error exponent is the same as that for the classical hypothesis testing problem of $X_1, \dots, X_n$ i.i.d.$\sim p$ versus $X_1, \dots, X_n$ i.i.d.$\sim q^*$ (see Corollary 3.4.6 in~\cite{dembo-zeitouni}). That is, for large values of $n$, the adversarial hypothesis testing game is not much different from the above classical setting (whose parameters are derived from the adversarial setting) in terms of classification error. We emphasize that we have not used any property of the specific structure of the mixed strategy Nash equilibrium in obtaining the error exponent associated with the classification error, and hence Theorem~\ref{theorem:errorexponent} is valid for any NE. We believe that obtaining the actual structure of a NE is a difficult problem, as the strategy spaces are infinite, and the utility functions do not possess any monotonicity properties in general. For numerical computation of error exponents in a simple case, see Section~\ref{section:numericals}.

We conclude this section by discussing the role of Assumptions~\ref{assm:a2} and~\ref{assm:a3}. We used \ref{assm:a3} to obtain the concentration of equilibrium in Lemma~\ref{lemma:eqconcentration}. Without this assumption, Theorem~\ref{theorem:errorexponent} is not valid; see Section~\ref{section:numericals} for numerical counter-examples. Also, in our setting, unlike the classical minimax testing, it is not clear whether it is always true that the error goes to $0$ as the number of samples becomes large, and whether the attacker should always play a point close to $q^*$ at equilibrium. It could be that playing a point far from $q^*$ is better if she can compensate the loss from $c$ from the error term. In fact, that is what happens when (A4) is not satisfied, since there is partial overlap of the decision region of the defender with the set $Q$. Regarding \ref{assm:a2}, when $c$ has multiple minimizers, our analysis can only tell us that the equilibrium of the attacker is supported around the set of minimizers for large enough $n$; to study error exponents in such cases, one has to do a finer analysis of characterizing the attacker's equilibrium. All in all, using \ref{assm:a2} and~\ref{assm:a3} allows us to establish interesting concentration properties of the equilibrium (which is not a priori clear) and error exponents associated with classification error without characterizing a NE, hence we believe that these assumptions serve as a good starting point. %Without making any assumptions on the cost function $c$ and the set of possible alternatives $Q$, we believe that one cannot expect to get a closed form expression for the error exponent (in fact, it is not even clear whether the limit exists in the first place).

\section{Numerical Experiments}\label{section:numericals}

%
%\begin{figure*}[t!]
%    \centering
%    \begin{subfigure}[t]{0.24\textwidth}
%        \centering
%        \includegraphics[scale=0.29]{bayesian_plots/fig1.png}
%        \caption{$n = 200$}
%        \label{fig:bayesian_plots/fig1}
%    \end{subfigure}%
%    ~ 
%    \begin{subfigure}[t]{0.24\textwidth}
%        \centering
%        \includegraphics[scale=0.29]{bayesian_plots/fig2.png}
%        \caption{$n = 250$}
%		\label{fig:bayesian_plots/fig2}
%    \end{subfigure}
%    \begin{subfigure}[t]{0.24\textwidth}
%        \centering
%        \includegraphics[scale=0.29]{bayesian_plots/fig7.png}
%        \caption{$Q =[0.7,0.9] , c(q) =|q-0.8|$}
%        \label{fig:bayesian_plots/fig7}
%    \end{subfigure}%
%    ~ 
%    \begin{subfigure}[t]{0.24\textwidth}
%        \centering
%        \includegraphics[scale=0.29]{bayesian_plots/fig9.png}
%        \caption{$Q =[0.6,0.9] , c(q) =3|q-0.9|$}
%        \label{fig:bayesian_plots/fig9}
%    \end{subfigure}
%    
%    
%    \caption{Best response plots for $c(q) = |q-0.8|$ (Figures (a) and (b)), and error exponents ((c) and (d))}
%\label{fig:main-figure}
%\end{figure*}

In this section, we present two numerical examples in the Bayesian formulation to illustrate the result in Theorem~\ref{theorem:errorexponent} and the importance of Assumption~\ref{assm:a3}. Due to space limitations, additional experiments in the Bayesian formulation are relegated to  \ifthenelse{\boolean{full}}{Appendix~\ref{section:additional-experiments}}{Appendix~C of the full version \cite{Sarath19aFull}}, which illustrate (a) best response strategies of the players, (b) existence of pure strategy Nash equilibrium for large values on $n$ as suggested by Lemma~\ref{lemma:eqconcentration}, and (c) importance of Assumption~\ref{assm:a3}.\footnote{\ifthenelse{\boolean{full}}{Appendix~\ref{section:additional-experiments}}{Appendix~C of the full version \cite{Sarath19aFull}} also contains numerical experiments in the Neyman-Pearson formulation presented in \ifthenelse{\boolean{full}}{Appendix~\ref{section:neyman-pearson}}{Appendix~A of the full version \cite{Sarath19aFull}}. The code used for our simulations is available at \url{https://github.com/sarath1789/ahtg_neurips2019}.} %,  (d) results in the Neyman-Pearson formulation.

We illustrate the result in Theorem~\ref{theorem:errorexponent} numerically in the following setting. We fix $\mathcal{X} = \{0,1\}$ (i.e. $d=2$) and each probability distribution on $\mathcal{X}$  is represented by the probability that it assigns to the symbol $1$, and hence $M_1(\mathcal{X})$ is viewed as the unit interval. We fix $p = 0.5$. For numerical computations, we discretize the set $Q$ into $100$ equally spaced points, and we only consider deterministic threshold-based decision rules for the defender. To compute a NE, we solve the linear programs associated with attacker as well as the defender for the zero-sum equivalent game of  $\mathcal{G}^B(2,n)$.

For the function $c(q) = |q-q^*|$ with $q^* = 0.8$, Figure~\ref{fig:error_exponents}(\subref{fig:bayesian_plots/fig7}) shows the error exponent at the NE computed by the above procedure as a function of the number of samples, from $n = 10$ to $n = 300$ in steps of $10$. As suggested by Theorem~\ref{theorem:errorexponent}, we see that the error exponents approach the value $\Lambda_0^*(0) = 0.054$ (the boundary of the decision region is  around the point $q =0.66 $, and $D(q||p) \approx D(q||q^*) \approx 0.054$).

\begin{figure*}[t!]
    \centering
    \begin{subfigure}[t]{0.50\textwidth}
        \centering
        \includegraphics[scale=0.40]{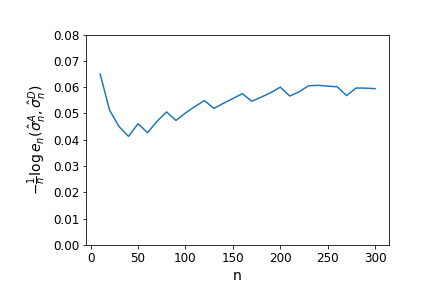}
        \caption{$Q =[0.7,0.9] , c(q) =|q-0.8|$}
        \label{fig:bayesian_plots/fig7}
    \end{subfigure}%
    ~ 
    \begin{subfigure}[t]{0.50\textwidth}
        \centering
        \includegraphics[scale=0.40]{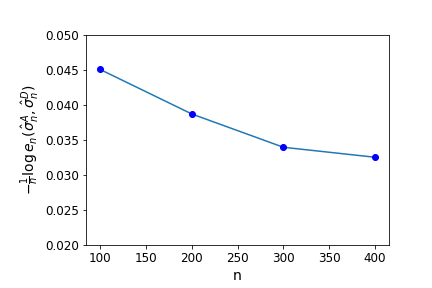}
        \caption{$Q =[0.6,0.9] , c(q) =3|q-0.9|$}
        \label{fig:bayesian_plots/fig9}
    \end{subfigure}
    \caption{Error exponents as a function of $n$}
\label{fig:error_exponents}
\end{figure*}

We now consider an example which demonstrates that, the result on error exponent in Theorem~\ref{theorem:errorexponent} may not be valid if Assumption~\ref{assm:a3} is not satisfied. In this experiment, we consider the case where $Q = [0.6, 0.9]$ and $q^* = 0.9$. Note that the present setting does not satisfy Assumption~\ref{assm:a3}. Figure~\ref{fig:error_exponents}(\subref{fig:bayesian_plots/fig9}) shows the error exponent at the equilibrium as a function of $n$, from $n=100$ to $n=400$ in steps of $100$, for the cost function $c(q) = 3|q-q^*|$. From this plot, we see that, the error exponents converge to somewhere around $0.032$, whereas $\Lambda^*_0(0)\approx 0.111$.

\section{Concluding Remarks}\label{section:conclusion}

In this paper, we studied hypothesis testing games that arise in the context of adversarial classification. We showed that, at equilibrium, the strategy of the classifier is to use a likelihood ratio test. We also examined the exponential rate of decay of classification error at equilibrium and showed that it is same as that of a classical testing problem with parameters derived from the adversarial model.

Throughout the paper, we assumed that the alphabet $\mathcal{X}$ is finite. This is a reasonable assumption in applications that deal with digital signals such as image forensics (an important application for adversarial hypothesis testing); and it is also a good starting point because even in this case, our analysis of the error exponents is nontrivial. Making $\mathcal{X}$ countable/uncountable will make the space $M_1(\mathcal{X})$ infinite dimensional, and the analysis of error exponents will become more difficult (e.g., the continuity of relative entropy is no longer true in this case, which we crucially use in our analysis), but the case of general state space $\mathcal{X}$ is an interesting future direction.

Finding the exact structure of the equilibrium for our hypothesis testing games is a challenging future direction. This will also shed some light on the error exponent analysis for the case when Assumption~\ref{assm:a3} is not satisfied. Another interesting future direction is to examine the hypothesis testing game in the sequential detection context where the defender can also decide the number of data samples for classification. In such a setting, an important question is to understand whether the optimal strategy of the classifier is to use a standard sequential probability ratio test.

\subsubsection*{Acknowledgments}
The first author is partially supported by the Cisco-IISc Research Fellowship grant. The work of the second author was supported in part by the French National Research Agency (ANR) through the ``Investissements d’avenir''
program (ANR-15-IDEX-02) and through grant ANR-16-
TERC0012; and by the Alexander von Humboldt Foundation. 
% Use unnumbered third level headings for the acknowledgments. All acknowledgments
% go at the end of the paper. Do not include acknowledgments in the anonymized
% submission, only in the final paper.

%\section*{References}

\bibliographystyle{apalike}
\medskip
\small
\bibliography{main.bib}

\ifthenelse{\boolean{full}}{

\appendix

\section{Hypothesis Testing Game: Neyman-Pearson Formulation}
\label{section:neyman-pearson}
In this section, we study the Neyman-Pearson version of the hypothesis testing problem. The presentation of results in this section is similar to the Bayesian case,  and we will also use the same notation as in the Bayesian formulation.
\subsection{Problem Formulation}
In the Neyman-Pearson point of view, we do not assume any knowledge on the probability that the external agent is an attacker. Fix $\varepsilon > 0$. As before, the strategy set of the attacker is the set $Q$. For the defender, motivated by the Neyman-Pearson approach for the classical hypothesis testing problem, we define the strategy set as the set of all randomized decision rules on $n$-length words whose false alarm probability is at most $\varepsilon$, i.e., 
\begin{align*}
\Phi_n = \{\varphi : \mathcal{X}^n \to [0,1]: P_n^{FA}(\varphi) \leq \varepsilon \},
\end{align*}
where $P_n^{FA}(\varphi) = \sum_{\mathbf{x}^n} \varphi(\mathbf{x}^n) p(\mathbf{x}^n)$ denotes the false alarm probability under the decision rule $\varphi$.

We now define the utilities. Similar to the Bayesian framework, the utility of the attacker is defined as
\begin{align}
u^A_n(q,\varphi) = \sum_{\mathbf{x}^n} (1-\varphi(\mathbf{x}^n)) q(\mathbf{x}^n) - c(q).
\label{eqn:attacker_revenue_pure_np}
\end{align}
Since we have already constrained the  strategy set of the defender by imposing an upper bound on the Type I error, the utility of the defender is defined as
\begin{align}
u^D_n(q,\varphi) = -\left(\sum_{\mathbf{x}^n} (1-\varphi(\mathbf{x}^n)) q(\mathbf{x}^n )\right),
\label{eqn:defender_revenue_pure_np}
\end{align}
which captures the Type II error.

We denote our two-player Neyman-Pearson hypothesis testing game with utility functions~(\ref{eqn:attacker_revenue_pure_np}) and (\ref{eqn:defender_revenue_pure_np}) by $\mathcal{G}^{NP}(\varepsilon, d,n)$. Similar to the Bayesian case, we will use Assumptions~\ref{assm:a1}-\ref{assm:a3} throughout the analysis of the Neyman-Pearson case. 
\subsection{Mixed Strategy Nash Equilibrium for $\mathcal{G}^{NP}(\varepsilon, d, n)$}
We now examine the mixed strategy Nash equilibrium for the game $\mathcal{G}^{NP}(\varepsilon, d, n)$. Similar to the Bayesian framework, we endow the set $Q$ with the standard Euclidean topology on $\mathbb{R}^d$ and the set $\Phi_n$ with the sup-norm metric. The following lemma asserts compactness of the strategy space of the defender. %, whose proof can be found in Appendix~1.
\begin{lemma}
The set $\Phi_n$ equipped with the metric $d_n$ is a compact metric space. \label{lemma:compactness-phin-np}
\end{lemma}

We also have joint continuity of the utility functions of both the players, and it can be proved similar to Lemma~\ref{lemma:cts-payoffs}.
\begin{lemma}Assume~\ref{assm:a1}-\ref{assm:a2}. Then, the utility functions $u_n^A$ and $u_d^D$ are continuous on $Q \times \Phi_n$. \label{lemma:cts-payoffs-np}
\end{lemma}

We are now ready to show the existence of a mixed strategy Nash equilibrium.
\begin{proposition}Assume~\ref{assm:a1}-\ref{assm:a2}. Then, there exists a mixed strategy Nash equilibrium for $\mathcal{G}^{NP}(\varepsilon,d,n)$. If $(\hat{\sigma}^A_n, \hat{\sigma}^D_n)$ is a NE, then $\hat{\sigma}^D_n$ is the point mass at the most powerful $\varepsilon$-level Neyman-Pearson test for $X_1,\dots, X_n\, i.i.d.\sim p$ versus $(X_1, \dots, X_n) \sim q_{\hat{\sigma}^A_n}$. \label{prop:defenderpure-np}
\end{proposition}
The existence of a mixed strategy equilibrium easily follows from the compactness of strategy spaces and continuity of utility functions. To show the specific form of defender's equilibrium strategy, we appeal to the Neyman-Pearson lemma.

Let $(\hat{\sigma}^A_n, \hat{\varphi}_n)$ denote a NE given by Proposition~\ref{prop:defenderpure-np}.
\subsection{Characterization of Equilibrium in the Binary case}
Consider the game $\mathcal{G}^{NP}(\varepsilon,d,n)$ in the binary case, i.e., $d=2$. Here, there are some interesting monotonicity properties of the utility functions that allow us to get a pure strategy Nash equilibrium for $\mathcal{G}^{NP}(\varepsilon,2,n)$, in which the defender plays a threshold-based test, i.e., declares the presence of an adversary whenever the number of 1's in the observation exceeds a threshold: % We have the following lemma, whose proof can be found in Appendix~1.
\begin{lemma}Assume~\ref{assm:a1}-\ref{assm:a2}. Then,  the defender admits a strictly dominant strategy, and there exists a pure strategy Nash equilibrium for $\mathcal{G}^{NP}(\varepsilon,2,n)$.
\label{lemma:pureeqforbinary_np}
\end{lemma}
\begin{remark} The monotonicity alluded to above is a consequence of the fact that $u^D_n$ captures just the Type II error. In the Bayesian framework, we do not have this monotonicity in $u^D_n$, due to the presence of both the Type I and Type II errors in $u^D_n$, and hence, existence of a pure strategy Nash equilibrium in the binary case cannot be guaranteed in the Bayesian framework.
\end{remark}
\subsection{Concentration Properties of Equilibrium}
In this section, we study some concentration properties of the equilibrium. We have the following two lemmas, which can be proved similar to the corresponding Lemmas for the Bayesian formulation in Section~\ref{section:eqconcentration}.
\begin{lemma}Assume~\ref{assm:a1}-\ref{assm:a2}. Let $(\hat{\sigma}^A_n, \hat{\varphi}_n)_{n \geq 1}$  be a sequence such that, for each $n \geq 1$,  $(\hat{\sigma}^A_n, \hat{\varphi}_n)$ is a mixed strategy Nash equilibrium for $\mathcal{G}^{NP}(\varepsilon,d,n)$. Then, $e_n(\hat{\sigma}^A_n, \hat{\varphi}_n) \to 0 $ as $n \to \infty$. \label{lemma:errorbound-np}
\end{lemma}
\begin{lemma}
Assume~\ref{assm:a1}-\ref{assm:a2}, and let $(\hat{\sigma}^A_n, \hat{\varphi}_n)_{n \geq 1}$  be as in Lemma~\ref{lemma:errorbound-np}. Then $\hat{\sigma}^A_n \to \delta_{q^*}$ weakly as $n \to \infty$. \label{lemma:weakconv-np}
\end{lemma}

We also have that the error at equilibrium goes to $0$ even when the attacker deviates from her equilibrium strategy.
\begin{lemma}\label{lemma:errorgoesto0-np}
Assume~\ref{assm:a1}-\ref{assm:a3}, and let $(\hat{\sigma}^A_n, \hat{\varphi}_n)_{n \geq 1}$ be as in Lemma~\ref{lemma:errorbound-np}. Then,
\begin{align*}
\sup_{q \in Q}e_n(q, \hat{\varphi}_n) \to 0 \text{ as }n \to \infty.
\end{align*}%
\end{lemma}
The main idea in the proof of the above lemma is to show that the acceptance region of $H_0$ under any equilibrium does not intersect the set $Q$. With this lemma at hand, we now have the following concentration property for the support of the equilibrium strategy of the attacker, which can be proved similar to Lemma~\ref{lemma:eqconcentration} in the Bayesian formulation.
\begin{lemma}
Assume~\ref{assm:a1}-\ref{assm:a3}, and let $(\hat{\sigma}^A_n, \hat{\varphi}_n)_{n \geq 1}$ be as in Lemma~\ref{lemma:errorbound-np}. Let $(q_n)_{n \geq 1}$ be a sequence such that $q_n \in supp(\hat{\sigma}^A_n)$ for each $n \geq 1$. Then, $q_n \to q^*$ as $n \to \infty$. \label{lemma:qnconvergence-np}
\end{lemma}
\begin{remark}
Note that, in one dimension ($d=1$), the acceptance region of an optimal Neyman-Pearson test for a fixed alternative $q$ will be a ``vanishingly small neighborhood of the null distribution $p$'' and that while it can still intersect $Q$ for finite $n$, it may not for large-enough $n$; so that Lemma~\ref{lemma:errorgoesto0-np} may always hold. However, it is unclear how this might generalize to higher dimension. Intuitively, in higher dimension, the acceptance region may become close to $p$ only from certain directions. We also note that our proof of Lemma A.6 actually uses Assumption (A4) and not a weaker version of it---see the expression of $\Gamma_n$ in the proof of Lemma~\ref{lemma:errorgoesto0-np}. Therefore, we believe that (A4) is needed in higher dimensions even for the Neyman-Pearson case; although it is possible that a weaker assumption will suffice in one dimension.
\end{remark}
\subsection{Error Exponents}
Our main result in the Neyman-Pearson formulation is the following theorem.
\begin{theorem}\label{theorem:errorexponent_np}
Assume~\ref{assm:a1}-\ref{assm:a3}, and let $(\hat{\sigma}^A_n, \hat{\varphi}_n)_{n \geq 1}$ be as in Lemma~\ref{lemma:errorbound-np}. Then,
\begin{align*}
\lim_{n \to \infty} \frac{1}{n} \log e_n(\hat{\sigma}^A_n, \hat{\varphi}_n) = - D(p||q^*).
\end{align*}
\end{theorem}
%See Appendix~1 for the proof.
%
Again, we note that the error exponent is the same as that of the classical Neyman-Pearson hypothesis testing problem between $p$ and $q^*$.

\section{Proofs}\label{section:proofs}
\subsection{Proof of Lemma~\ref{lemma:cts-payoffs}}
Since we are on a metric space, it suffices to show sequential continuity. Let $\{(q_k,\varphi_k), k \geq 1\}$ be a sequence such that $(q_k,\varphi_k) \to (q, \varphi)$ as $k \to \infty$. First, consider $u^D_n$. Notice that, for each $\mathbf{x}^n$, we have $q_k(\mathbf{x}^n) \to q(\mathbf{x}^n)$, and $\varphi_k (\mathbf{x}^n) \to \varphi(\mathbf{x}^n)$ as $k \to \infty$. Therefore, we have that $q_k(\mathbf{x}^n) \varphi_k(\mathbf{x}^n) \to q(\mathbf{x}^n) \varphi(\mathbf{x}^n)$, which yields that,
\begin{align*}
\lim_{k \to \infty} \sum_{\mathbf{x}^n} q_k(\mathbf{x}^n) \varphi_k(\mathbf{x}^n) = \sum_{\mathbf{x}^n} q(\mathbf{x}^n) \varphi(\mathbf{x}^n).
\end{align*}
Similarly, we also have
\begin{align*}
\lim_{k \to \infty} \sum_{\mathbf{x}^n}p(\mathbf{x}^n) \varphi_k(\mathbf{x}^n) = \sum_{\mathbf{x}^n} p(\mathbf{x}^n) \varphi(\mathbf{x}^n).
\end{align*}

Therefore, we have that $u^D_n(q_k,\varphi_k) \to u^D_n(q,\varphi)$ as $k \to \infty$ which proves continuity of the utility of the defender. Using similar arguments, and by using the continuity of the cost function $c$ on $Q$, we see that $u^A_n(q_k,\varphi_k) \to u^A_n(q,\varphi)$ as $k \to \infty$, which shows the continuity of the utility of the attacker. 
\qed

\subsection{Proof of Proposition~\ref{prop:defenderpure}}
$\mathcal{G}^B(d,n)$ is a two-player game with compact strategy spaces. Also, by Lemma~\ref{lemma:cts-payoffs}, the utilities (in pure strategies) of both the attacker and the defender are jointly continuous on $Q \times \Phi_n$. Therefore, an application of the Glicksberg fixed point theorem (see, for example, Corollary~2.4 in~\cite{reny-05}) tell us that there exists a mixed strategy Nash equilibrium (NE) for the adversarial hypothesis testing game $\mathcal{G}^B(d,n)$.

We now show the structure of the equilibrium strategy of the defender. Note that, for any $\varphi \in \Phi_n$,  $u^D_n(\hat{\sigma}^A_n,\varphi) = -1 + \sum_{\mathbf{x}^n} \varphi(\mathbf{x}^n) (q_{\hat{\sigma}^A_n}(\mathbf{x}^n ) - \gamma p(\mathbf{x}^n)), $ where $q_{\hat{\sigma}^A_n}(\mathbf{x}^n) = \int q(\mathbf{x}^n) \hat{\sigma}^A_n(dq)$. Therefore, using the characterization of a NE (see Proposition 140.1 in~\cite{osborne}), it follows that for any $\varphi \in \text{ supp}(\hat{\sigma}^{D}_n)$, we have
\begin{align*}
\varphi(\mathbf{x}^n) = \left\{
\begin{array}{ll}
1, & \text{if } q_{\hat{\sigma}^A_n}(\mathbf{x}^n) - \gamma p(\mathbf{x}^n) > 0,\\
0, & \text{if } q_{\hat{\sigma}^A_n}(\mathbf{x}^n) - \gamma p(\mathbf{x}^n) < 0.
\end{array}
\right.
\end{align*}

Now, define $\hat{\varphi}_n$ such that $\hat{\varphi}_n(\mathbf{x}^n) = \int \varphi(\mathbf{x}^n)\hat{\sigma}^{D}_n(d\varphi)$ whenever $\mathbf{x}^n$ is such that $q_{\hat{\sigma}^A_n}(\mathbf{x}^n) - \gamma p(\mathbf{x}^n) = 0$, and that satisfies the above condition when $\mathbf{x}^n$ is such that $q_{\hat{\sigma}^A_n}(\mathbf{x}^n) - \gamma p(\mathbf{x}^n) \neq 0$. Consider the strategy profile $(\hat{\sigma}^A_n,\hat{\varphi}_n)$ where the defender plays the pure strategy $\hat{\varphi}_n$. By the choice of $\hat{\varphi}_n$, we see that $u^A_n(q,\hat{\varphi}_n) = u^A_n(q,\hat{\sigma}^D_n)$ for all $q \in Q$, and $u^D_n(\hat{\sigma}^A_n,\varphi) \leq u^D_n(\hat{\sigma}^A_n,\hat{\varphi}_n)$ for any $\varphi \in \Phi_n$. Therefore, using the characterization of a NE, we see that $(\hat{\sigma}^A_n,\hat{\varphi}_n)$ is a NE. This completes the proof of the Proposition. 
\qed
\subsection{Proof of Lemma~\ref{lemma:errorbound}}
By Assumption~\ref{assm:a1}, there exist a $\delta >0$ such that $B(p,\delta) \cap Q = \emptyset$, where $B(p,\delta)$ denotes an open ball of radius $\delta$ centered at $p$. Let $\varphi^{\delta}$ denote the deterministic decision rule whose rejection region is the set $B(p,\delta)^c$, i.e., $\varphi^{\delta}(\mathbf{x}^n) = 1$ whenever $\mathcal{P}_{\mathbf{x}^n} \in B(p,\delta)^c$ and $\varphi^{\delta}(\mathbf{x}^n) = 0$ otherwise, where $\mathcal{P}_{\mathbf{x}^n} \in M_1(\mathcal{X})$ denotes the type of $\mathbf{x}^n$, i.e., $\mathcal{P}_{\mathbf{x}^n}(i) = \frac{1}{n}\sum_{k=1}^n 1_{\{x_k = i\}}$. Since $(\hat{\sigma}^A_n, \hat{\sigma}^D_n)$ is a Nash equilibrium, and $e_n(\hat{\sigma}^A_n, \hat{\sigma}^D_n) = -u^D_n (\hat{\sigma}^A_n, \hat{\sigma}^D_n)$, we see that
\begin{align}
e_n(\hat{\sigma}^A_n, \hat{\sigma}^D_n) \leq e_n(\hat{\sigma}^A_n, \varphi^{\delta}),
\label{eqn:bound-error}
\end{align}
where $(\hat{\sigma}^A_n, \varphi^{\delta})$ denotes the strategy profile where the attacker plays the mixed strategy $\hat{\sigma}^A_n$ and the defender plays the pure strategy $\varphi^{\delta}$.

We now proceed to bound the error term $e_n(\hat{\sigma}^A_n, \varphi^{\delta})$. We have
\begin{align*}
e_n (\hat{\sigma}^A_n, \varphi^{\delta}) &= \int \left[\sum_{\mathbf{x}^n} (1-\varphi^{\delta}(\mathbf{x}^n))q(\mathbf{x}^n) + \gamma \varphi^{\delta}(\mathbf{x}^n)p(\mathbf{x}^n)\right] \hat{\sigma}^A_n (dq) \\
& =   \int q(\mathcal{P}_{\mathbf{x}^n} \in B(p,\delta)) \hat{\sigma}^A_n (dq) +  \gamma p(\mathcal{P}_{\mathbf{x}^n} \in B(p,\delta)^c).
\end{align*}
We bound the first term above using a simple upper bound for the probability of observing a given type under a given distribution (see, for example, Lemma 2.1.9 in~\cite{dembo-zeitouni}). Let $\mathcal{P}_n$ denote the set of all possible types of an $n$-length word. For any $q \in Q$, we have that 
\begin{align*}
q(\mathcal{P}_{\mathbf{x}^n} \in B(p,\delta)) & \leq \sum_{\nu \in B(p,\delta) \cap \mathcal{P}_n} e^{-nD(\nu || q)} \\
& \leq |B(p,\delta) \cap \mathcal{P}_n| \ e^{-n \inf_{\nu \in B(p,\delta)} D(\nu||q)} \\
& \leq (n+1)^d \ e^{-n \inf_{\nu \in B(p,\delta), q \in Q} D(\nu||q)},
\end{align*}
where the last inequality follows since $|\mathcal{P}_n| \leq (n+1)^d$. Therefore,
\begin{align*}
e_n(\hat{\sigma}^A_n, \varphi^{\delta}) \leq & (n+1)^d \ e^{-n \inf_{\nu \in B(p,\delta), q \in Q} D(\nu||q)}  + \gamma p(\mathcal{P}_{\mathbf{x}^n} \in B(p,\delta)^c).
\end{align*}
The first term above goes to $0$ as $n\to \infty$, since $\inf_{\nu \in B(p,\delta), q \in Q} D(\nu||q) > 0$. Also, by the weak law of large numbers, we see that $\mathcal{P}_{\mathbf{x}^n}$ converges to $p$ in probability under the null hypothesis $H_0$. Therefore,
\begin{align*}
p(\mathcal{P}_{\mathbf{x}^n} \in B(p,\delta)^c) \to 0.
\end{align*}
Hence, we conclude that $e_n(\hat{\sigma}^A_n,\varphi^{\delta}) \to 0$ as $n \to \infty$. Combining this with~(\ref{eqn:bound-error}) completes the proof of the Lemma. 
\qed
\subsection{Proof of Lemma~\ref{lemma:weakconv}}
From Lemma~\ref{lemma:errorbound}, we have $e_n(\hat{\sigma}^A_n, \hat{\sigma}^D_n) \to 0$ as $n \to \infty$. Since u$^A_n(q^*, \hat{\sigma}^D_n) \geq -c(q^*)$ and since $(\hat{\sigma}^A_n, \hat{\sigma}^D_n)$ is a NE for all $n \geq 1$, it follows that 
\begin{align*}
\int c(q) \hat{\sigma}^A_n (dq) \to c(q^*)
\end{align*} 
as $n \to \infty$.
Since $(\hat{\sigma}^A_n)_{n \geq 1}$ is a sequence of probability measures on the compact space $Q$, by Prohorov's theorem (see Theorem 1, Section 2 in Chapter 3 of~\cite{shiryaev}), there exists a weakly convergent subsequence (say $(n_k)_{k \geq 1})$. Let $\mu$ denote the weak limit of $(\hat{\sigma}^A_{n_k})_{k \geq 1}$. Then, we have,
\begin{align}
c(q^*) & = \lim_{k \to \infty} \int c(q) \hat{\sigma}^A_{n_k} (dq) \nonumber  \\
& = \int c(q) \mu(dq), \label{eqn:convintofcq}
\end{align}
where the last equality follows from weak convergence. 

We now claim that $\mu = \delta_{q^*}$. Suppose not. Then, there exists $\varepsilon >0$ such that $\mu(B(q^*,\varepsilon)^c) > 0$. By Assumption~\ref{assm:a2}, for the above $\varepsilon$, there exists a $\delta>0$ such that $c(q) >  c(q^*) + \delta$ whenever $q \in B(q^*,\varepsilon)^c$. Therefore,
\begin{align*}
\int_{Q} c(q) &\mu(dq) = \int_{B(q^*,\varepsilon)} c(q) \mu(dq) + \int_{B(q^*,\varepsilon)^c}c(q) \mu(dq) \\
& \geq c(q^*) \mu(B(q^*,\varepsilon)) + (c(q^*)+\delta) \mu(B(q^*,\varepsilon)^c) \\
& = c(q^*) + \delta \mu(B(q^*,\varepsilon)^c) \\
& > c(q^*),
\end{align*}
which contradicts~(\ref{eqn:convintofcq}). Therefore, it follows that $\mu(B(q^*, \varepsilon)^c) = 0$ for every $\varepsilon >0$ and hence $\mu = \delta_{q^*}$. Since $\mu$ is independent of the subsequence $(n_k)_{k \geq 1}$, it follows that the whole sequence $(\hat{\sigma}^A_n)_{n \geq 1}$ converges to $\delta_{q^*}$ (see Lemma 1, Section 3 in Chapter 3 of~\cite{shiryaev}). This completes the proof of the lemma. 
\qed

To prove Lemma~\ref{lemma:errorgoesto0}, we need the following lemma, which asserts uniform convergence of integrals of the relative entropy functional w.r.t. the equilibrium strategy of the attacker.
\begin{lemma}
Let $(\hat{\sigma}^A_n, \hat{\sigma}^D_n)_{n \geq 1}$ be as in Lemma~\ref{lemma:errorbound}. Then,
\begin{align*}
\sup_{\mu \in M_1(\mathcal{X})} \left| \int D(\mu ||q) \hat{\sigma}^A_n(dq) - D(\mu||q^*) \right| \to 0 \text{ as } n \to \infty.
\end{align*}
\label{lemma:uconv-integrals}
\end{lemma}
\begin{proof}
Fix $\varepsilon > 0$ and $ \mu \in M_1(\mathcal{X})$. Then, using the uniform continuity of the relative entropy function on $M_1(\mathcal{X}) \times Q$, there exists $\delta >0$ such that
\begin{align*}
\sup_{q \in Q} |D(\mu ||q) - D(\mu^\prime ||q) | < \varepsilon \text{ for all } \mu^\prime \in B(\mu,\delta).
\end{align*}
Therefore, for all $\mu^\prime \in B(\mu,\delta)$, we have
\begin{align*}
 \left| \int  D(\mu^\prime ||q) \hat{\sigma}^A_n(dq)- \int D(\mu||q)  \hat{\sigma}^A_n(dq)\right|  & \leq \int |D(\mu^\prime ||q) - D(\mu||q) | \hat{\sigma}^A_n(dq) \\
& \leq \varepsilon, \text{ for all } n \geq 1.
\end{align*}
Also, using weak convergence of $(\hat{\sigma}^A_n)_{n \geq 1}$ to the point mass at $q^*$, there exists $N_ \mu \geq 1$ such that
\begin{align*}
\left|\int D(\mu ||q) \hat{\sigma}^A_n(dq) - D(\mu||q^*) \right| \leq \varepsilon \text{ for all } n \geq N_{\mu}.
\end{align*}

Note that, the sets $B(\mu,\delta)_{\mu \in Q}$ is an open cover for $M_1(\mathcal{X})$. By compactness of the space $M_1(\mathcal{X})$, extract a finite subcover $B(\mu_i, \delta), 1 \leq i \leq k$. Put $N = \max\{N_{\mu_1}, \dots, N_{\mu_k}\}.$ Then, for all $n \geq N$, we have
\begin{align*}
 \left| \int D(\mu || q) \hat{\sigma}^A_n(dq) - D(\mu ||q^*) \right| &\leq \left| \int D(\mu || q) \hat{\sigma}^A_n(dq)  - \int D(\mu_i || q) \hat{\sigma}^A_n(dq)\right| + \\ & \ \ \ \ \left|  \int D(\mu_i || q) \hat{\sigma}^A_n(dq) - D(\mu_i || q^*) \right| + \\ & \ \ \ \ \left| D(\mu_i || q^*)  - D(\mu ||q^*) \right| \\
& \leq 3 \varepsilon,
\end{align*}
where $\mu_i$ is such that $\mu \in B(\mu_i,\delta)$. The result now follows since $\varepsilon$ and $\mu$ are arbitrary.
\end{proof}

\subsection{Proof of Lemma~\ref{lemma:errorgoesto0}}
Recall the decision rule $\hat{\varphi}_n$ from Proposition~\ref{prop:defenderpure}. Note that if $H_0$ is accepted under $\hat{\varphi}_n$ when the defender observes $\mathbf{x}^n$, then we have
\begin{align*}
\frac{\int q(\mathbf{x}^n)\hat{\sigma}^A_n(dq)}{p(\mathbf{x}^n)} \leq \gamma
\end{align*}
(note that there could be randomization when equality holds above). By Proposition~\ref{prop:defenderpure}, notice that $(\hat{\sigma}^A_n, \hat{\varphi}_n)$ is a Nash equilibrium, and $e_n(q, \hat{\sigma}^D_n) =  e_n(q, \hat{\varphi}_n)$ for all $q \in Q$. Therefore it suffices to show that $\sup_{q \in Q} e_n(q, \hat{\varphi}_n) \to 0$ as $n \to \infty$.

Note that, the acceptance region of $H_0$ under the decision rule $\hat{\varphi}_n$ is type-based, i.e., for every $n$-length word $\mathbf{x}^n$,  $\hat{\varphi}_n(\mathbf{x}^n)$ depends only on $\mathcal{P}_{\mathbf{x}^n}$. Therefore, if $H_0$ is accepted when the defender observes $\mathbf{x}^n$, the type $\mathcal{P}_{\mathbf{x}^n}$ must belong to the following subset of $M_1(\mathcal{X})$:
\begin{align*}
\left\{ \mathcal{P}_{\mathbf{x}^n} : \log \int \frac{q(\mathbf{x}^n)}{p(\mathbf{x}^n)} \hat{\sigma}^A_n (dq) \leq \log \gamma \right\}.
\end{align*}
Define 
\begin{align*}
\Gamma_n = \left\{\mathcal{P}_{\mathbf{x}^n} : \int \log\left( \frac{q(\mathbf{x}^n)}{p(\mathbf{x}^n)} \right) \hat{\sigma}^A_n (dq) \leq \log \gamma \right\}.
\end{align*}
Notice that, by Jensen's inequality, the acceptance region of $H_0$ under the decision rule $\hat{\varphi}_n$ is a subset of the above set $\Gamma_n$. Also, it is easy to check that,
\begin{align*}
\Gamma_n = \{ \mu \in M_1(\mathcal{X}):D( \mu || p) - & \int D(\mu || q) \hat{\sigma}^A_n(dq)  \leq \frac{\log \gamma}{n} \} \cap \mathcal{P}_n.
\end{align*}
We now show that the set $\Gamma_n$ does not intersect the set $Q$ for large enough $n$. First, notice that, the set $\{\mu \in M_1(\mathcal{X}): D(\mu||p) \leq D(\mu||q^*)\}$ is closed in $M_1(\mathcal{X})$. Therefore, by Assumption~\ref{assm:a3}, there exists $\eta	 > 0$ such that $Q^\eta \cap \{\mu \in M_1(\mathcal{X}): D(\mu||p) \leq D(\mu||q^*)\} = \emptyset$, where $Q^\eta  = \{\mu \in M_1(\mathcal{X}): \inf_{q \in Q} || \mu -q || \leq \eta\}$ is the $\eta$-expansion of the set $Q$.

We show that there exists $N \geq 1$ such that $Q^\eta \cap \Gamma_n = \emptyset$ for all $n \geq N$. Suppose not, then we can find a sequence $(\mu_n)_{n \geq 1}$ such that $\mu_n \in Q^\eta$ and $\mu_n\in \Gamma_n$ for all $n \geq 1$. Since $Q^\eta$ is compact, we can find a subsequence $(n_k)_{k \geq 1}$ along which $\mu_n$ converges, and let $\mu = \lim_{k \to \infty} \mu_{n_k} \in Q^\eta.$ Since $\mu_n \in \Gamma_n$ for all $n \geq 1$, using Lemma~\ref{lemma:uconv-integrals}, we see that $\mu$ satisfies
$D(\mu || p) \leq D(\mu ||q^*)$. This contradicts the fact that $Q^\eta \cap \Gamma_n = \emptyset $, and hence, there exists $N \geq 1$ such that $Q^\eta \cap \Gamma_n = \emptyset$ for all $n \geq N$.

By the law of large numbers, we have
\begin{align*}
\sup_{q \in Q} q(\mathcal{P}_{\mathbf{x}^n} \notin B(q, \eta)) \to 0,
\end{align*}
and 
\begin{align*}
p(\mathcal{P}_{\mathbf{x}^n} \notin \Gamma_n) \to 0
\end{align*}
as $n \to \infty$. But, notice that
\begin{align*}
e_n(q, \hat{\varphi}_n) &\leq q(\mathcal{P}_{\mathbf{x}^n} \in \Gamma_n) + \gamma p(\mathcal{P}_{\mathbf{x}^n} \notin \Gamma_n) \\
& \leq  q(\mathcal{P}_{\mathbf{x}^n} \notin B(q, \eta)) + \gamma p(\mathcal{P}_{\mathbf{x}^n} \notin \Gamma_n)
\end{align*}
for all $q \in Q$ and $n \geq N$. Therefore,
\begin{align*}
\sup_{q \in Q} e_n(q, \hat{\varphi}_n) & \leq \sup_{q \in Q} q(\mathcal{P}_{\mathbf{x}^n} \notin B(q, \eta)) + \gamma p(\mathcal{P}_{\mathbf{x}^n} \notin \Gamma_n) \\
& \to 0
\end{align*}  
as $n \to  \infty$. 
\qed

\subsection{Proof of Lemma~\ref{lemma:eqconcentration}}
Fix $\varepsilon > 0$. By Lemma~\ref{lemma:errorgoesto0}, there exists $N_{\varepsilon}$ such that
\begin{align*}
e_n(q_n, \hat{\sigma}^D_n) \leq \varepsilon
\end{align*}
for all $n \geq N_{\varepsilon}$. Therefore, 
\begin{align*}
u^A_n(q_n, \hat{\sigma}^D_n) \leq \varepsilon - c(q_n)
\end{align*}
for all $n \geq N_{\varepsilon}$. However, by playing the pure strategy $q^*$, the attacker utility is 
\begin{align*}
u^A_n(q^* ,\hat{\sigma}^D_n) \geq -c(q^*)
\end{align*}
for all $n \geq 1$. Since $(\hat{\sigma}^A_n, \hat{\sigma}^D_n)$ is a Nash equilibrium, and since $q_n \in supp(\hat{\sigma}^A_n)$, we must have $u^A_n(q_n, \hat{\sigma}^D_n) \geq  u^A_n(q^*, \hat{\sigma}^D_n)$ for all $n \geq N_{\varepsilon}$. That is,
\begin{align*}
c(q_n) \leq c(q) + \varepsilon
\end{align*}
for all $n \geq N_{\varepsilon}$. Thus, it follows that, $c(q_n) \to c(q^*)$ as $n \to \infty$. Using the definition of $q^*$, we see that $q_n \to q^*$ as $n \to \infty$. 
\qed

\subsection{Proof of Theorem~\ref{theorem:errorexponent}}
First, we obtain the asymptotic lower bound. Towards this, we shall consider an equivalent zero-sum game for $\mathcal{G}^B(d,n)$. For $q \in Q$ and $\varphi \in \Phi_n$, define
\begin{align*}
u^{eq}_n(q,\varphi)& = \sum_{\mathbf{x}^n} (1-\varphi(\mathbf{x}^n)) q(\mathbf{x}^n) +\gamma  \sum_{\mathbf{x}^n}\varphi(\mathbf{x}^n)p(\mathbf{x}^n) - c(q).
\end{align*}
Observe that, as far as the attacker is concerned, for any $\varphi \in \Phi_n$, maximizing $u^A_n(\cdot, \varphi)$ is the same as maximizing $u^{eq}_n(\cdot, \varphi)$, as the extra term  present in $u^{eq}_n$ does not depend on the attacker strategy. Similarly, for any $q \in Q$, maximizing the defender's utility function $u^D_n(q,\cdot)$ is the same as minimizing $u^{eq}_n(q,\cdot)$, as the cost function $c$ does not depend on the defender's strategy. Therefore, we see that $\mathcal{G}^B(d,n)$ is best-response equivalent to a two-player zero sum game (with attacker being first player and defender being second player) with the above utility for the first player. Hence, any equilibrium for the original game is also going to be an equilibrium for the zero-sum equivalent game with utility function $u^{eq}_n $ (see Definition~4~in~\cite{dritsoula-etal-17} and the remarks before Theorem 2).

Consider the strategy profile $(q^*, \hat{\sigma}^D_n)$, i.e., the attacker plays the pure strategy $q^*$ and the defender plays the mixed strategy $\hat{\sigma}^D_n$ that comes from the equilibrium. By definition of the Nash equilibrium, and the equivalence of $\mathcal{G}^B(d,n)$ with the above zero-sum game, we have 
\begin{align}
u_n^{eq} (\hat{\sigma}^A_n, \hat{\sigma}^D_n) \geq u^{eq}_n (q^*, \hat{\sigma}^D_n). \label{eqn:utilityeqbound}
\end{align}
($u_n^{eq} (\hat{\sigma}^A_n, \hat{\sigma}^D_n)$ denotes the utility in mixed extension of the equivalent zero-sum game).

Define the deterministic decision rule $\varphi^*_n$ by
\begin{align*}
\varphi^*_n(\mathbf{x}^n) = \left\{
\begin{array}{ll}
1, & \text{ if } \frac{q^*(\mathbf{x}^n)}{p(\mathbf{x}^n)} \geq \gamma, \\
0, & \text{ otherwise}.
\end{array}
\right.
\end{align*}

It is easy to see that $\varphi^*_n$ minimizes $e_n(q^*, \cdot)$. Writing the probabilities $p(\mathbf{x}^n)$ and $q(\mathbf{x}^n)$ in terms of $\mathcal{P}_{\mathbf{x}^n}$, it is easy to check that, the acceptance region of $\varphi^*_n$ is given by
\begin{align*}
\Gamma^*_n = \{\nu \in M_1(\mathcal{X})\cap \mathcal{P}_n : D(\nu||q^*) - D(\nu|| p) > \frac{\log \gamma}{n}\},
\end{align*}
i.e., $\varphi^*_n(\mathbf{x}^n) = 0$ whenever $\mathcal{P}_{\mathbf{x}^n} \in \Gamma^*_n$, and $\varphi^*_n(\mathbf{x}^n) = 1$ otherwise.
Noting that $u_n^{eq} (q,\varphi) = e_n (q,\varphi) - c(q)$, (\ref{eqn:utilityeqbound}) becomes
\begin{align}
e_n(\hat{\sigma}^A_n, \hat{\sigma}^D_n) & \geq \int \sum_{\mathbf{x}^n} \left( (1-\varphi(\mathbf{x}^n)) q^*(\mathbf{x}^n) +\gamma \varphi(\mathbf{x}^n)p(\mathbf{x}^n) \right)\hat{\sigma}^D_n(d\varphi)  - c(q^*) + \int c(q) \hat{\sigma}^A_n(dq) \nonumber \\
& \geq \int \sum_{\mathbf{x}^n} \left( (1-\varphi(\mathbf{x}^n)) q^*(\mathbf{x}^n) +\gamma \varphi(\mathbf{x}^n)p(\mathbf{x}^n) \right)\hat{\sigma}^D_n(d\varphi)  \nonumber \\
& \geq \sum_{\mathbf{x}^n} \left( (1-\varphi^*_n(\mathbf{x}^n)) q^*(\mathbf{x}^n) +\gamma \varphi^*_n(\mathbf{x}^n)p(\mathbf{x}^n) \right), \label{eqn:errorboundlower}
\end{align}
where the second inequality follows from the definition of $q^*$, and the last inequality follows from the optimality of $\varphi^*_n$. The quantitiy in the RHS of the last inequality is the minimum Bayesian error for the following standard binary hypothesis testing problem: under the null hypothesis, each symbol in $\mathbf{x}^n$ is generated independently from $p$,  and under the alternate hypothesis, each symbol is generated independently from $q^*$. It is well known that (see, for example, Corollary 3.4.6 in~\cite{dembo-zeitouni}),
\begin{align*}
\liminf_{n \to \infty} \frac{1}{n} \log e_n(q^*,\varphi^*_n) \geq -\Lambda_0^*(0),
\end{align*}
and hence, from (\ref{eqn:utilityeqbound}) and (\ref{eqn:errorboundlower}), it follows that,
\begin{align}
\liminf_{n \to \infty} \frac{1}{n} \log e_n(\hat{\sigma}^A_n, \hat{\sigma}^D_n)  \geq -\Lambda^*_0(0). \label{eqn:ratelowerbound}
\end{align}

We now proceed to show the upper bound. Define the decision rule $\varphi^\prime_n$ by
\begin{align*}
\varphi^\prime_n (\mathbf{x}^n) = \left\{
\begin{array}{ll}
1, & \text{ if } \frac{q^*(\mathbf{x}^n)}{p(\mathbf{x}^n)} \geq 1, \\
0, & \text{ otherwise}.
\end{array}
\right.
\end{align*}
Similar to the decision rule $\varphi^*_n$, the acceptance region of $\varphi^\prime_n$ can be written as
\begin{align*}
\Gamma^\prime = \{\nu \in M_1(\mathcal{X}) : D(\nu||q^*) - D(\nu|| p) > 0\},
\end{align*}
i.e., $\varphi^\prime_n(\mathbf{x}^n) = 0$ if $\mathcal{P}_{\mathbf{x}^n} \in \Gamma^\prime$, and $\varphi^\prime_n(\mathbf{x}^n) = 1$ otherwise. By the definition of a Nash equilibrium, and noting that $u^D_n(\hat{\sigma}^A_n, \hat{\sigma}^D_n) = - e_n(\hat{\sigma}^A_n, \hat{\sigma}^D_n)$, we have
\begin{align}
e_n(\hat{\sigma}^A_n, \hat{\sigma}^D_n) \leq e_n(\hat{\sigma}^A_n, \varphi_n^\prime), \label{eqn:rateupperphiprime}
\end{align}
where $(\hat{\sigma}^A_n, \varphi_n^\prime)$ denotes the strategy profile where the attacker plays the mixed strategy $\hat{\sigma}^A_n$ that comes form the equilibrium, and the defender plays the pure strategy $\varphi^\prime_n$. We have,
\begin{align*}
e_n(\hat{\sigma}^A_n, \varphi_n^\prime) &= \int \sum_{\mathbf{x}^n} \left( (1-\varphi^\prime_n(\mathbf{x}^n)) q(\mathbf{x}^n) +\gamma \varphi^\prime_n(\mathbf{x}^n)p(\mathbf{x}^n) \right) \sigma^A_n(dq) \\
& = \int q(\mathcal{P}_{\mathbf{x}^n} \in \Gamma^\prime) \hat{\sigma}^A_n(dq) + p(\mathcal{P}_{\mathbf{x}^n} \in (\Gamma^\prime)^c).
\end{align*}
Consider the first term. Using the upper bound on the probability of observing a type under a given distribution (Lemma~2.1.9 in~\cite{dembo-zeitouni}), we have
\begin{align*}
 q(\mathcal{P}_{\mathbf{x}^n} \in \Gamma^\prime) \leq (n+1)^d e^{-n\inf_{\nu \in \Gamma^\prime}D(\nu || q)}.
\end{align*}
Fix $\varepsilon > 0$.  Since  the relative entropy is jointly uniformly continuous on $\Gamma^\prime \times Q$,  there exists a $\delta >0$ such that
\begin{align*}
D(\nu || q) \geq D(\nu || q^*) -\varepsilon \end{align*}
for all $\nu \in \Gamma^\prime$ whenever $ ||q -q^*||_2 < \delta$. For the above $\delta$, by Lemma~\ref{lemma:eqconcentration}, there exists $N_{\delta}$ such that $||q - q^*||_2 < \delta$ whenever $q \in \text{supp}(\hat{\sigma}^A_n)$ for all $n \geq N_{\delta}$. Therefore, we see that, for all $n \geq N_{\delta}$ and $\nu \in \Gamma^\prime$,
\begin{align*}
D(\nu || q) \geq D(\nu || q^*) -\varepsilon \text{ for all } q \in \text{supp}(\hat{\sigma}^A_n).
\end{align*}
Therefore, for all $n \geq N_{\delta}$, we have
\begin{align*}
q(\mathcal{P}_{\mathbf{x}^n} \in \Gamma^\prime) \leq (n+1)^d e^{-n(\inf_{\nu \in \Gamma^\prime}D(\nu || q^*)-\varepsilon)}
\end{align*}
for all $ q\in \text{supp}(\hat{\sigma}^A_n)$. For the second term, using Lemma~2.1.9 in~\cite{dembo-zeitouni}, we have
\begin{align*}
    p(\mathcal{P}_{\mathbf{x}^n} \in (\Gamma^\prime)^c) \leq e^{-n\inf_{\nu \in (\Gamma^\prime)^c}D(\nu || p)}.
\end{align*}
It can be easily shown that (for example, see Exercise 3.4.14(b) in \cite{dembo-zeitouni}), $\inf_{\nu \in \Gamma^\prime}D(\nu || q^*) = \inf_{\nu \in (\Gamma^\prime)^c}D(\nu || p) =  \Lambda_0^*(0)$. Hence, the above implies that
\begin{align*}
\limsup_{n \to \infty} \frac{1}{n} \log e_n(\hat{\sigma}^A_n, \varphi_n^\prime) \leq -\Lambda_0^*(0)+\varepsilon.
\end{align*}
Letting $\varepsilon \to 0$, we get
\begin{align*}
\limsup_{n \to \infty} \frac{1}{n} \log e_n(\hat{\sigma}^A_n, \varphi_n^\prime) \leq -\Lambda_0^*(0).
\end{align*}
 Therefore, from~(\ref{eqn:rateupperphiprime}) and the above inequality, we have
\begin{align}
\limsup_{n \to \infty} \frac{1}{n} \log e_n(\hat{\sigma}^A_n, \hat{\sigma}^D_n) \leq -\Lambda_0^*(0). \label{eqn:rateupperbound}
\end{align}
The theorem now follows from~(\ref{eqn:ratelowerbound}) and (\ref{eqn:rateupperbound}). 
\qed

\subsection{Proof of Lemma~\ref{lemma:compactness-phin-np}}
We show sequential compactness of $\Phi_n$. Let $(\varphi_k)_{k \geq 1}$ be a sequence in $\Phi_n$. Let $\mathbf{x}^n_1, \dots, \mathbf{x}^n_{2^n}$ denote the elements of $\mathcal{X}^n$. Since $\varphi_{n}(\mathbf{x}^n) \in [0,1]$ for all $\mathbf{x}^n \in \mathcal{X}^n$, there exists a subsequence $(k^{(1)}_l)_{l \geq 1}$ along which $\varphi(\mathbf{x}^n_1)$ converges. We can then extract a further subsequence $(k^{(2)}_l)_{l \geq 1}$ of $(k^{(1)}_l)_{l \geq 1}$ along which $\varphi(\mathbf{x}^n_2)$ converges. Repeating the above procedure $2^n$ times, we see that, there exists a subsequence $(k_l)_{l \geq 1}$ along which $\varphi(\mathbf{x}^n)$  converges for all $\mathbf{x}^n \in \mathcal{X}^n$. Put
\begin{align*}
\varphi(\mathbf{x}^n) = \lim_{l \to \infty} \varphi_{k_l}(\mathbf{x}^n), \, \, \mathbf{x}^n \in \mathcal{X}^n.
\end{align*}
It is clear that $d_n(\varphi_{k_l}, \varphi) \to 0$ as $l \to \infty$, and we have
\begin{align*}
P^{FA}_n(\varphi) & = \sum_{\mathbf{x}^n} \varphi(\mathbf{x}^n) p(\mathbf{x}^n) \\
& = \sum_{\mathbf{x}^n} \lim_{l \to \infty } (\varphi_{k_l}(\mathbf{x}^n)) p(\mathbf{x}^n) \\
& = \lim_{l \to \infty} \sum_{\mathbf{x}^n} \varphi_{k_l}(\mathbf{x}^n)p(\mathbf{x}^n) \\
& = \lim_{l \to \infty} P^{FA}(\varphi_{k_l})\\
& \leq \varepsilon,
\end{align*}
since $P^{FA}(\varphi_{k_l}) \leq \varepsilon$ for all $l \geq 1$. This shows that the space $\Phi_n$ equipped with the metric $d_n$ is sequentially compact, and hence compact.
\qed

\subsection{Proof of Proposition~\ref{prop:defenderpure-np}}
By Lemma~\ref{lemma:compactness-phin-np}, the strategy space of the defender is compact. Also, the strategy space of the attacker is compact under the standard Euclidean topology on $\mathbb{R}^d$. By Lemma~\ref{lemma:cts-payoffs-np}, we see that the utility functions of both players are jointly continuous. Therefore, by the Glicksberg fixed point theorem (see, for example, Corollary 2.4. in~\cite{reny-05}), there exists a mixed strategy Nash equilibrium for the game $\mathcal{G}(\varepsilon, d, n)$.

We now show the structure of the equilibrium strategy of the defender. Let $(\hat{\sigma}^A_n, \hat{\sigma}^D_n)$ denote a mixed strategy Nash equilibrium of $\mathcal{G}(\varepsilon,d,n)$. By the property of Nash equilibrium, we have that $e_n(\hat{\sigma}^A_n, \varphi) = e_n(\hat{\sigma}^A_n, \hat{\sigma}^D_n)$ for all $\varphi \in supp(\hat{\sigma}^D_n)$. We claim that $P^{FA}(\varphi) = \varepsilon$ for all $\varphi \in supp(\hat{\sigma}^D_n$). If there exists $\varphi \in supp(\hat{\sigma}^D_n)$ with $P^{FA}(\varphi) < \varepsilon$, then we can find $\mathbf{x}^n_0 \in \mathcal{X}^n$ such that $\varphi(\mathbf{x}^n_0) = 0$ and $\delta >0$ such that the decision rule defined by $\varphi^\prime(\mathbf{x}^n) =\varphi(\mathbf{x}^n)$ for all $\mathbf{x}^n \neq \mathbf{x}^n_0$, and $\varphi^\prime(\mathbf{x}^n_0) = \delta$ has the property that $P^{FA}_n(\varphi^\prime) \leq \varepsilon$ and $e_n(\hat{\sigma}^A_n,\varphi^\prime) < e_n(\hat{\sigma}^A_n, \hat{\sigma}^D_n)$. This contradicts the fact that, $(\hat{\sigma}^A_n, \hat{\sigma}^D_n)$ is a Nash equilibrium, which proves our claim.

But, note that
\begin{align*}
e_n(\hat{\sigma}^A_n, \hat{\sigma}^D_n) & = \int \sum_{\mathbf{x}^n} (1-\varphi(\mathbf{x}^n)) q(\mathbf{x}^n) \hat{\sigma}^A_n(dq) \hat{\sigma}^D_n(d\varphi) \\
& =  \sum_{\mathbf{x}^n}  \left[  \int (1-\varphi(\mathbf{x}^n)) q(\mathbf{x}^n) \hat{\sigma}^A_n(dq) \hat{\sigma}^D_n(d \varphi) \right] \\ 
& = \sum_{\mathbf{x}^n}  \left[  \int (1-\varphi(\mathbf{x}^n)) q_{\hat{\sigma}^A_n}(\mathbf{x}^n)) \hat{\sigma}^D_n(d \varphi) \right]
\end{align*}
where $q_{\hat{\sigma}^A_n} \in M_1(\mathcal{X}^n)$ is given by
\begin{align*}
q_{\hat{\sigma}^A_n} (\mathbf{x}^n) = \int q(\mathbf{x}^n) \hat{\sigma}^A_n(dq).
\end{align*}
That is, when the attacker plays the Nash equilibrium $\hat{\sigma}^A_n$, the defender faces the problem of distinguishing between the two alternatives: (i) $(X_1, \dots, X_n)$ is generated by i.i.d. $p$, versus (ii) $(X_1, \dots, X_n)$ is generated by $q_{\hat{\sigma}^A_n}$. By the Neyman-Pearson lemma, we know that there exists a Neyman-Pearson decision rule $\hat{\varphi}_n \in \Phi_n$ with the property that $P^{FA}(\hat{\varphi}_n) = \varepsilon$ and $e_n(\hat{\sigma}^A_n, \cdot)$ is minimized by $\hat{\varphi}_n$ on $\Phi_n$. Since every $\varphi \in supp(\hat{\sigma}^D_n)$ minimizes $e_n(\hat{\sigma}^A_n, \cdot)$, and $P^{FA}(\varphi) = \varepsilon$,	and since each $\mathbf{x}^n \in \mathcal{X}^n$ has positive probability of observing under both $H_0$ and $H_1$, an application of the uniqueness part in Neyman-Pearson lemma (see, for example, Section 5.1 in \cite{ferguson-67}) yields that that $\hat{\sigma}^D_n = \delta_{\hat{\varphi}_n}$. This completes the proof. 
\qed

\subsection{Proof of Lemma~\ref{lemma:pureeqforbinary_np}}
Recall the definition of a Neyman-Pearson decision rule. In the binary case, since the comparison of the ratio $\frac{q(\mathbf{x}^n)}{p(\mathbf{x}^n)}$ to a threshold is the same as comparison of the number of $1$'s in the $n$-length word $\mathbf{x}^n$ to some other threshold, we see that any Neyman-Pearson decision rule $\varphi$ must necessarily be of the following form:
\begin{align}
\varphi(\mathbf{x}^n) = \left\{
\begin{array}{ll}
0, &\text{ if } \mathcal{P}_{\mathbf{x}^n}(1) \in \{0, \frac{1}{n} , \dots, \frac{k}{n}\}, \\
\pi, & \text{ if } \mathcal{P}_{\mathbf{x}^n}(1) = \frac{k+1}{n}, \\
1, & \text{ if } \mathcal{P}_{\mathbf{x}^n}(1) \in \{\frac{k+2}{n} , \dots, 1\},
\end{array}
\right.
\label{eqn:nprule_1d}
\end{align}
for some $\pi \in [0,1]$ and $0 \leq k \leq n$. Here, $\mathcal{P}_{\mathbf{x}^n}(1)$ denotes the fraction of $1$'s in $\mathbf{x}^n$. The false alarm probability of the above decision rule is
\begin{align*}
P^{FA}_n(\varphi) & = p\left(\mathcal{P}_{\mathbf{x}^n}(1) \in \{0, \frac{1}{n} , \dots, \frac{k}{n}\} \right) +\pi p\left(\mathcal{P}_{\mathbf{x}^n}(1) = \frac{k+1}{n}\right).
\end{align*}
Since every $n$-length word $\mathbf{x}^n$ has positive probability under the distribution $p$, we see that, there exists a unique $k$ and $\pi$ such that $P^{FA}_n(\varphi) = \varepsilon$. Let $\hat{\varphi}_n$ denote the above Neyman-Pearson decision rule. Then, by the Neyman-Pearson lemma (see, for example, Proposition II.D.1~in~\cite{poor}), we see that,
\begin{align*}
\hat{\varphi}_n = \arg \max_{\varphi \in \Phi_n} u^D_n(q, \varphi) \text{ for all } q \in Q.
\end{align*}
Thus, the defender has a unique strictly dominant strategy. Using the continuity of $c$, and the continuity of the Type II error term in the attacker's strategy, we see that $u^A_n(\cdot, \hat{\varphi}_n)$ is continuous on Q, and hence there exist a maximum. Therefore, letting the attacker play a pure strategy $\hat{q}_n$ that maximizes $u^A_n(\cdot, \hat{\varphi}_n)$ yields a pure strategy Nash equilibrium $(\hat{q}_n, \hat{\varphi}_n)$.
\qed

To prove Lemma~\ref{lemma:errorgoesto0-np}, we need the following lemma, which can be proved similar to Lemma~\ref{lemma:uconv-integrals}.

\begin{lemma}
Let $(\hat{\sigma}^A_n, \hat{\varphi}_n)_{n \geq 1}$ be as in Lemma~\ref{lemma:errorbound-np}. Then,
\begin{align*}
\sup_{\mu \in M_1(\mathcal{X})} \left| \int D(\mu ||q) \hat{\sigma}^A_n(dq) - D(\mu||q^*) \right| \to 0 \text{ as } n \to \infty.
\end{align*}
\label{lemma:uconv-integrals-np}
\end{lemma}

\subsection{Proof of Lemma~\ref{lemma:errorgoesto0-np}}
Let $\gamma_n$ denote the threshold and $\pi$ denote the randomization used in the decision rule $\hat{\varphi}_n$, i.e., $\hat{\varphi}_n$ is of the form
\begin{align*}
\hat{\varphi}_n(\mathbf{x}^n) = \left\{
\begin{array}{ll}
1, \text{ if } & \frac{q_{\hat{\sigma}^A_n}(\mathbf{x}^n)}{p(\mathbf{x}^n)} > \gamma_n, \\
\pi, \text{ if } & \frac{q_{\hat{\sigma}^A_n}(\mathbf{x}^n)}{p(\mathbf{x}^n)} = \gamma_n, \\
0, \text{ if } & \frac{q_{\hat{\sigma}^A_n}(\mathbf{x}^n)}{p(\mathbf{x}^n)} < \gamma_n.
\end{array}
\right.
\end{align*}

We first claim that $\limsup_{n \to \infty} \gamma_n \leq 1$. Since $P^{FA}(\hat{\varphi}_n) = \varepsilon$, we have that
\begin{align}
p\left(  \frac{q_{\hat{\sigma}^A_n}(\mathbf{x}^n)}{p(\mathbf{x}^n)} \geq \gamma_n \right) \geq \varepsilon \label{eqn:pfaequalsepsilon-np}
\end{align}
But, using the probability of observing an $n$-length word under a distribution (see, for example, Lemma 2.1.6 in~\cite{dembo-zeitouni}), we have
\begin{align*}
q_{\hat{\sigma}^A_n(\mathbf{x}^n)} &= \int q(\mathbf{x}^n) \hat{\sigma}^A_n(dq) \\
& = \int e^{-n(H(\mathcal{P}_{\mathbf{x}^n}) + D(\mathcal{P}_{\mathbf{x}^n}||q ))} \hat{\sigma}^A_n(dq),
\end{align*}
and
\begin{align*}
p(\mathbf{x}^n) = e^{-n(H(\mathcal{P}_{\mathbf{x}^n}) + D(\mathcal{P}_{\mathbf{x}^n})||p )}.
\end{align*}
Therefore,
\begin{align*}
p\left(  \frac{q_{\hat{\sigma}^A_n}(\mathbf{x}^n)}{p(\mathbf{x}^n)} \geq \gamma_n\right) &= p\left(  \int e^{n(D(\mathcal{P}_{\mathbf{x}^n}||p) -D(\mathcal{P}_{\mathbf{x}^n}||q))} \geq \gamma_n\right) \\
&\leq p\left( e^{n(D(\mathcal{P}_{\mathbf{x}^n}||p) -\inf_{q \in Q}D(\mathcal{P}_{\mathbf{x}^n}||q))} \geq \gamma_n \right)
\end{align*}
By Assumption~\ref{assm:a1}, we can choose $\delta >0 $ such that $D(\mu || p) < \inf_{q \in Q} D(\mu ||q)$ for all $\mu \in B(p,\delta)$. Thus,
\begin{align*}
&p\left( e^{n(D(\mathcal{P}_{\mathbf{x}^n}||p) -\inf_{q \in Q}D(\mathcal{P}_{\mathbf{x}^n}||q))} \geq \gamma_n \right) \\
&= p\left( e^{n(D(\mathcal{P}_{\mathbf{x}^n}||p) -\inf_{q \in Q}D(\mathcal{P}_{\mathbf{x}^n}||q))} \geq \gamma_n, \mathcal{P}_{\mathbf{x}^n} \in B(p,\delta)\right) \\ &+ p\left( e^{n(D(\mathcal{P}_{\mathbf{x}^n}||p) -\inf_{q \in Q}D(\mathcal{P}_{\mathbf{x}^n}||q))} \geq \gamma_n, \mathcal{P}_{\mathbf{x}^n} \notin B(p,\delta)\right) 
\end{align*}
By law of large numbers, $p(\mathcal{P}_{\mathbf{x}^n} \notin B(p, \delta)) \to 0$, and hence the second term above goes to $0$ as $n \to \infty$. Suppose that $\limsup_{n \to \infty} \gamma_n > 1$, then there exists a subsequence $(n_k)_{k \geq 1}$ such that $\gamma_{n_k} > 1$ for all $k \geq 1$. Therefore, along this subsequence, the first term above becomes
\begin{align*}
&p( e^{n_k(D(\mathcal{P}_{\mathbf{x}^{n_k}}||p) -\inf_{q \in Q}D(\mathcal{P}_{\mathbf{x}^{n_k}}||q))} \geq \gamma_{n_k}, \mathcal{P}_{\mathbf{x}^{n_k}} \in B(p,\delta) ) \\ 
& \leq p( e^{n(D(\mathcal{P}_{\mathbf{x}^{n_k}}||p) -\inf_{q \in Q}D(\mathcal{P}_{\mathbf{x}^{n_k}}||q))}  >1, \mathcal{P}_{\mathbf{x}^{n_k}} \in B(p,\delta))
\end{align*}
which goes to $0$ as $n \to \infty$ by the choice of $\delta$. This implies that, $p\left(  \frac{q_{\hat{\sigma}^A_n}(\mathbf{x}^n)}{p(\mathbf{x}^n)} \geq \gamma_n\right) \to 0$ as $n \to \infty$, which contradicts~(\ref{eqn:pfaequalsepsilon-np}). Therefore, we must have $\limsup_{n \to \infty} \gamma_n \leq 1$.

We now argue that, for some $\eta > 0$, the acceptance set of $H_0$ under $\hat{\varphi}_n$ does not intersect the set $Q^\eta$. Towards this, consider the set
\begin{align*}
\Gamma_n = \left\{\mathcal{P}_{\mathbf{x}^n} : \int \log\left( \frac{q(\mathbf{x}^n)}{p(\mathbf{x}^n)} \right) \hat{\sigma}^A_n (dq) \leq \log \gamma_n \right\}.
\end{align*}
Notice that, by Jensen's inequality, the acceptance region of $H_0$ under the decision rule $\hat{\varphi}_n$ is a subset of the above set $\Gamma_n$. Also, it is easy to check that,
\begin{align*}
\Gamma_n = \{ \mu \in M_1(\mathcal{X}):D( \mu || p) - \int D(\mu || q) \hat{\sigma}^A_n(dq) 
\leq \frac{\log \gamma_n}{n} \} \cap \mathcal{P}_n.
\end{align*}
We now show that the set $\Gamma_n$ does not intersect the set $Q$ for large enough $n$. First, notice that, the set $\{\mu \in M_1(\mathcal{X}): D(\mu||p) \leq D(\mu||q^*)\}$ is closed in $M_1(\mathcal{X})$.  Therefore,by Assumption~\ref{assm:a3} there exists $\eta > 0$ such that $\{\mu \in M_1(\mathcal{X}): D(\mu||p) \leq D(\mu||q^*)\} \cap Q^\eta = \emptyset$.

We show that there exists $N \geq 1$ such that $Q^\eta \cap \Gamma_n = \emptyset$ for all $n \geq N$. Suppose not, then we can find a sequence $(\mu_n)_{n \geq 1}$ such that $\mu_n \in Q^\eta$ and $\mu_n\in \Gamma_n$ for all $n \geq 1$. Since $Q^\eta$ is compact, we can find a subsequence $(n_k)_{k \geq 1}$ along which $\mu_n$ converges, and let $\mu = \lim_{k \to \infty} \mu_{n_k} \in Q^\eta.$ Since $\mu_n \in \Gamma_n$ for all $n \geq 1$, using Lemma~\ref{lemma:uconv-integrals-np} and the fact that $\limsup_{n \to \infty} \gamma_n = 0$, we see that $\mu$ satisfies
$D(\mu || p) \leq D(\mu ||q^*)$. This contradicts the fact that $Q^\eta \cap \Gamma_n = \emptyset $, and hence, there exists $N \geq 1$ such that $Q^\eta \cap \Gamma_n = \emptyset$ for all $n \geq N$.

By the law of large numbers, we have
\begin{align*}
\sup_{q \in Q} q(\mathcal{P}_{\mathbf{x}^n} \notin B(q, \eta)) \to 0
\end{align*}
as $n \to \infty$. But, notice that
\begin{align*}
e_n(q, \hat{\varphi}_n) &\leq q(\mathcal{P}_{\mathbf{x}^n} \in \Gamma_n) \\
& \leq  q(\mathcal{P}_{\mathbf{x}^n} \notin B(q, \eta))
\end{align*}
for all $q \in Q$ and $n \geq N$. Therefore,
\begin{align*}
\sup_{q \in Q} e_n(q, \hat{\varphi}_n) \leq \sup_{q \in Q} q(\mathcal{P}_{\mathbf{x}^n} \notin B(q, \eta)) \to 0
\end{align*}  
as $n \to  \infty$. 
\qed

\subsection{Proof of Theorem~\ref{theorem:errorexponent_np}}
We proceed through similar steps as in the proof of Theorem~\ref{theorem:errorexponent}. To show the lower bound, we let the attacker play the pure strategy $q^*$ instead of the her equilibrium strategy $\hat{\sigma}^A_n$ for all $n\geq 1$. Since $u^A_n(q, \varphi) = e_n(q,\varphi) - c(q)$, and since $(\hat{\sigma}^A_n,\hat{\varphi}_n)$ is a Nash equilibrium for $\mathcal{G}^{NP}(\varepsilon,d,n)$, we see that
\begin{align*}
e_n(\hat{\sigma}^A_n,\hat{\varphi}_n) &\geq e_n(q^*, \hat{\varphi}_n) - c(q^*) + \int c(q) \hat{\sigma}^A_n(dq) \\
& \geq e_n(q^*, \hat{\varphi}_n) \\
& \geq e_n(q^*, \varphi^*_n)
\end{align*}
where $\varphi^*_n$ denotes the best $\varepsilon$-level Neyman-Pearson test for distinguishing $p$ from $q^*$ from $n$ independent samples. Here, the  second inequality follows from the definition of $q^*$, and the last inequallity follows from the optimality of Neyman-Pearson test $\varphi^*_n$. Hence, using Stein's lemma (see, for example, Lemma 3.4.7 in~\cite{dembo-zeitouni}), we see that
\begin{align}
\liminf_{n \to \infty} \frac{1}{n} \log e_n(\hat{\sigma}^A_n, \hat{\varphi}_n) \geq -D(p||q^*).
\label{eqn:ratelb-np}
\end{align}

We now show the upper bound. Fix $0<\delta<1$ such that $B(p,\delta)\cap Q = \emptyset$, and consider the deterministic decision rule $\varphi^{\delta}$ with acceptance region $B(p,\delta)$, i.e., $\varphi^\delta(\mathbf{x}^n) = 0$ whenever $\mathcal{P}_{\mathbf{x}^n} \in B(p,\delta)$ and $\varphi^\delta(\mathbf{x}^n) = 1$ otherwise. To obtain the upper bound, we let the defender play the strategy $\varphi^{\delta}_n$ for all $n \geq 1$. Since $(\hat{\sigma}^A_n, \hat{\varphi}_n)$ is a Nash equilibrium, and $u^D_n(q,\varphi) = -e_n(q,\varphi)$, we have
\begin{align}
e_n(\hat{\sigma}^A_n, \hat{\varphi}_n) &\leq e_n(\hat{\sigma}^A_n, \varphi^{\delta}) \nonumber \\
& = \int q(\mathcal{P}_{\mathbf{x}^n} \in \Gamma^{\delta}) \hat{\sigma}^A_n(dq) \nonumber \\
&\leq \int (n+1)^d e^{-n \inf_{\nu \in \Gamma^{\delta}} D(\nu || q)} \hat{\sigma}^A_n(dq), \label{eqn:errorub-np}
\end{align}
where the last inequality follows form the upper bound in Lemma~2.1.9 in~\cite{dembo-zeitouni}.
By Lemma~\ref{lemma:qnconvergence-np} and by the uniform continuity of $D(\cdot||\cdot)$ on $\Gamma^{\delta} \times Q$, there exists $N_{\delta} \geq 1$ such that
\begin{align*}
D(\nu || q) \geq D(\nu || q^*) - \delta \text{ for all } \nu \in \Gamma^{\delta}, q \in supp(\hat{\sigma}^A_n) \text{ and } n \geq N_{\delta}.
\end{align*}
Therefore, (\ref{eqn:errorub-np}) implies that
\begin{align*}
\limsup_{n \to \infty} \frac{1}{n} \log e_n(\hat{\sigma}^A_n, \hat{\varphi}_n) & \leq   - \inf_{\nu \in \Gamma^{\delta}} D(\nu || q^*)+ \delta.
\end{align*}
Letting $\delta \to 0$ and using the continuity of $D(\cdot||q^*)$ on $M_1(\mathcal{X})$, we get
\begin{align}
\limsup_{n \to \infty} \frac{1}{n} \log e_n(\hat{\sigma}^A_n, \hat{\varphi}_n) & \leq   - D(p || q^*).
\label{eqn:rateub-np}
\end{align}
The result now follows form~(\ref{eqn:ratelb-np}) and~(\ref{eqn:rateub-np}).
\qed

\section{Additional Numerical Experiments}\label{section:additional-experiments}
\subsection{Bayesian Formulation}
\label{appendix:numerical-experiments-bayesian}
\begin{figure*}[t!]
    \centering
    \begin{subfigure}[t]{0.50\textwidth}
        \centering
        \includegraphics[scale=0.40]{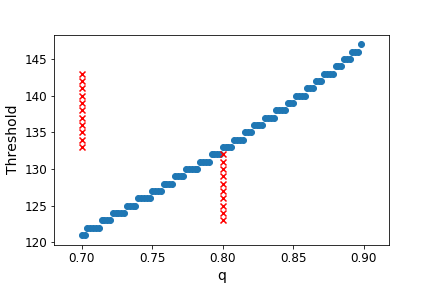}
        \caption{$n = 200$}
        \label{fig:bayesian_plots/fig1}
    \end{subfigure}%
    ~ 
    \begin{subfigure}[t]{0.50\textwidth}
        \centering
        \includegraphics[scale=0.40]{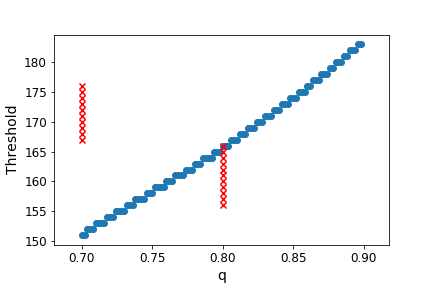}
        \caption{$n = 250$}
        \label{fig:bayesian_plots/fig2}
    \end{subfigure}
    \caption{Best response plots for $c(q) = |q-0.8|$ }
\label{fig:best_response_bayesian}
\end{figure*}

As explained in Section~\ref{section:numericals}, we fix $\mathcal{X} = \{0,1\}$ and $p = 0.5$. For numerical computations, we discretize the set $Q$ into $100$ equally spaced points, and we only consider deterministic threshold-based decision rules for the defender. 

We first examine the best response of the players. We fix $Q = [0.7,0.9]$ and the cost function to be $c(q) = |q-q^*|$ where $q^* = 0.8$. Figure~\ref{fig:best_response_bayesian}(\subref{fig:bayesian_plots/fig1}) shows the best response of the players for $n = 200$. The $x$-axis shows the strategy space of the attacker and $y$-axis shows the defender's threshold. The blue curve plots the best response of the defender, and the red curve plots the best response of the attacker for $20$ thresholds around the best response threshold corresponding to $q^*$. As we see from the figure, the two curves do not intersect (the best threshold for $0.8$ is $ 133$, whereas the best value of $q$ for threshold $133 $ is $0.7$) and hence this suggests that there is no pure strategy equilibrium in this case.  Figure~\ref{fig:best_response_bayesian}(\subref{fig:bayesian_plots/fig2}) plots the best response curves for $n=250$. We see that the two curves intersect (the point of intersection is when the attacker plays $0.8$ and the defender plays the threshold $166$).  However, this does not mean that there is a pure equilibrium, as our discretization may not capture the exact value of the attacker strategy. To see whether this is the case, we plot the attacker revenue when the defender plays the threshold $166$ over a finer grid around $q^*$ ($1000$ equally sized points on the interval of length $1/(100*n)$ around $q^*$), which is shown in Figure~\ref{fig:finer-plots-bayesian}(\subref{fig:bayesian_plots/fig3}). From this, we observe that the maximum of the attacker utility is indeed attained at the point $ q=0.8$. This suggests that there is a pure strategy Nash equilibrium when the attacker plays $0.8$ and defender plays the threshold $166$, though we could not prove this analytically.

\begin{figure*}[t!]
    \centering
    \begin{subfigure}[t]{0.50\textwidth}
        \centering
        \includegraphics[scale=0.40]{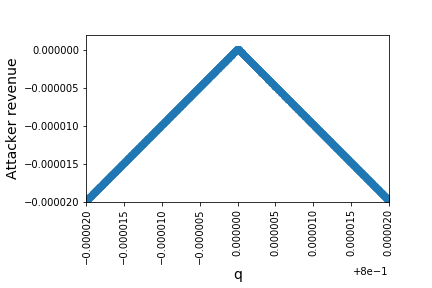}
        \caption{$c(q) = |q-0.8|, n=250$, defender plays threshold $166$}
        \label{fig:bayesian_plots/fig3}
    \end{subfigure}%
    ~ 
    \begin{subfigure}[t]{0.50\textwidth}
        \centering
        \includegraphics[scale=0.40]{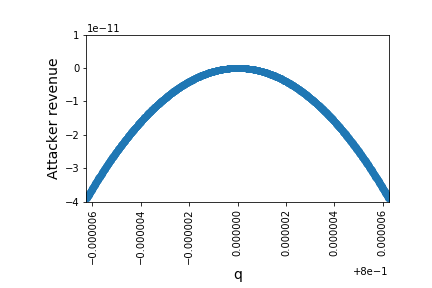}
        \caption{$c(q) = (q-0.8)^2, n=800$, defender plays threshold $529$}
        \label{fig:bayesian_plots/fig6}
    \end{subfigure}
    \caption{Finer plots of attacker revenue around $q^*$ for specific defender thresholds }
\label{fig:finer-plots-bayesian}
\end{figure*}

%
%\begin{figure*}[t!]
%    \centering
%    \begin{subfigure}[t]{0.5\textwidth}
%        \centering
%        \includegraphics[scale=0.35]{bayesian_plots/fig3.png}
%        \caption{$c(q) = |q-0.8|, n=250$, defender plays threshold $166$}
%        \label{fig:bayesian_plots/fig3}
%    \end{subfigure}%
%    ~ 
%    \begin{subfigure}[t]{0.5\textwidth}
%        \centering
%        \includegraphics[scale=0.35]{bayesian_plots/fig6.png}
%        \caption{$c(q) = (q-0.8)^2, n=800$, defender plays threshold $529$}
%		\label{fig:bayesian_plots/fig6}
%    \end{subfigure}
%    \caption{Finer plots of attacker revenue around $q^*$ for specific defender thresholds}
%\label{fig:finer-plots-bayesian}
%\end{figure*}

\begin{figure*}[t!]
    \centering
    \begin{subfigure}[t]{0.50\textwidth}
        \centering
        \includegraphics[scale=0.40]{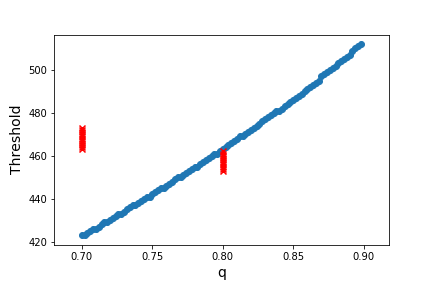}
        \caption{$n=700$}
        \label{fig:bayesian_plots/fig4}
    \end{subfigure}%
    ~ 
    \begin{subfigure}[t]{0.50\textwidth}
        \centering
        \includegraphics[scale=0.40]{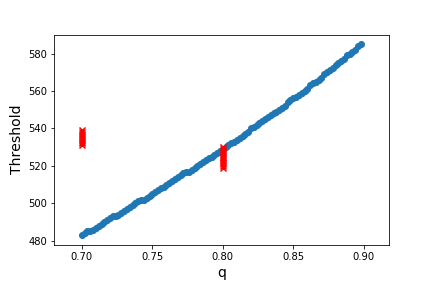}
        \caption{$n=800$}
        \label{fig:bayesian_plots/fig5}
    \end{subfigure}
    \caption{Best response plots for the cost function $c(q) = (q-0.8)^2$}
\label{fig:best-response-plots-appendix-bayesian}
\end{figure*}

%
%\begin{figure*}[t!]
%    \centering
%    \begin{subfigure}[t]{0.5\textwidth}
%        \centering
%        \includegraphics[height=1.2in]{bayesian_plots/fig4.png}
%        \caption{$n = 700$}
%        \label{fig:bayesian_plots/fig4}
%    \end{subfigure}%
%    ~ 
%    \begin{subfigure}[t]{0.5\textwidth}
%        \centering
%        \includegraphics[height=1.2in]{bayesian_plots/fig5.png}
%        \caption{$n = 800$}
%		\label{fig:bayesian_plots/fig5}
%    \end{subfigure}
%    \caption{Best response plots for the cost function $c(q) = (q-0.8)^2$}
%\label{fig:best-response-plots-appendix-bayesian}
%\end{figure*}

Similar to the best response plots in Figure~\ref{fig:best_response_bayesian}, we plot the best response plots for the quadratic cost function $c(q) = (q-q^*)^2 $ where $q^* = 0.8$. Figure~\ref{fig:best-response-plots-appendix-bayesian}(\subref{fig:bayesian_plots/fig4}) shows the best response plots for $n=700 $. Here, they don't intersect (the best threshold for $0.8$ is $463$ whereas the best value of $q$ for threshold $463$ is $0.7$), which shows that there is no pure equilibrium for the game. Figure~\ref{fig:best-response-plots-appendix-bayesian}(\subref{fig:bayesian_plots/fig5}) shows the best repose plots for $n=800$. Here, we see that the curves intersect (when the attacker plays $0.8$ and defender plays the threshold $529$). As before, Figure~\ref{fig:finer-plots-bayesian}(\subref{fig:bayesian_plots/fig6}) shows a finer plot of the utility of the attacker around the point $q^*$. We see that the utility of the attacker is indeed maximized at $0.8$, which suggests that there is a pure strategy equilibrium when the attacker plays the strategy $0.8$ and defender plays the threshold $529$. 

From these experiments for the linear as well as quadratic cost functions, as we expect, there is no incentive for the attacker to deviate much form the point $q^*$, since for large values of $n$, the error term in the utility of the attacker does not contribute to the overall revenue. However, in the second case, since the cost function has zero derivative at $q^*$, it is not clear whether a slight deviation form the point $q^*$ can increase the error term compared to the decrease in the cost function, so that the overall utility of the attacker increases. Therefore, the existence of a pure strategy Nash equilibrium with the attacker strategy being equal to $q^*$ is surprising in this case. However, in the first case, since the left and right derivatives of the cost function at $q^*$ are non-zero, the  decease in the cost function is much larger compared to the possible  increase in the error term as we slightly deviate from $q^*$, and hence it is reasonable to expect the existence of a pure strategy equilibrium at $q^*$ for large $n$.

Comparing the two cost functions, a much higher value of $n$ is needed in the second case for us to have a pure equilibrium at $q^*$, since the increase in the cost function is much slower in the second case as we move away from the point $q^*$.

We now give two more examples that does not satisfy Assumption~\ref{assm:a3} whose error exponents are different from Theorem~\ref{theorem:errorexponent}. As before,  $Q = [0.6, 0.9]$ and $q^* = 0.9$, and recall that Assumption~\ref{assm:a3} is not satisfied in this case. We consider the linear cost function $c(q) = 2|q-q^*|$ and the quadratic cost function $c(q) = (q-q^*)^2$. Figures~\ref{fig:error-exponents-appendix-bayesian}(\subref{fig:bayesian_plots/fig10}) and \ref{fig:error-exponents-appendix-bayesian}(\subref{fig:bayesian_plots/fig11}) show the error exponent at the equilibrium as a function of $n$, from $n=100$ to $n=400$ in steps of $100$. From these plots, we see that, the error exponents converge to somewhere around $0.023$ and $0.011$ respectively for the above two cost functions, whereas the value of $\Lambda^*_0(0)$ is around $0.111$.

\begin{figure*}[t!]
    \centering
    \begin{subfigure}[t]{0.50\textwidth}
        \centering
        \includegraphics[scale=0.40]{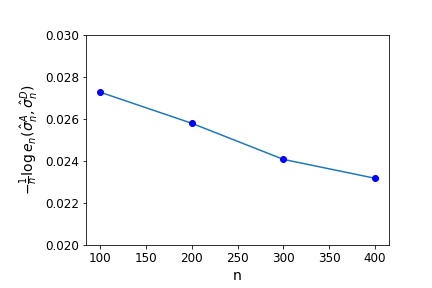}
        \caption{$c(q) = 2|q-0.9|$}
        \label{fig:bayesian_plots/fig10}
    \end{subfigure}%
    ~ 
    \begin{subfigure}[t]{0.50\textwidth}
        \centering
        \includegraphics[scale=0.40]{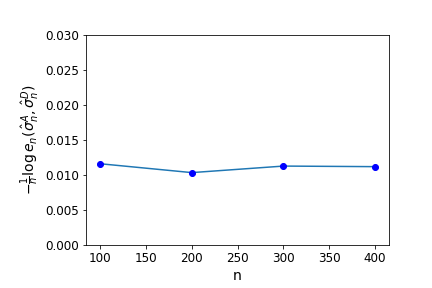}
        \caption{$c(q) = (q-0.9)^2$}
        \label{fig:bayesian_plots/fig11}
    \end{subfigure}
    \caption{Error exponents as a function of $n$}
\label{fig:error-exponents-appendix-bayesian}
\end{figure*}

%
%\begin{figure*}[t!]
%    \centering
%    \begin{subfigure}[t]{0.5\textwidth}
%        \centering
%        \includegraphics[scale=0.35]{bayesian_plots/fig10.png}
%        \caption{$c(q) = 2|q-0.9|$}
%        \label{fig:bayesian_plots/fig10}
%    \end{subfigure}%
%    ~ 
%    \begin{subfigure}[t]{0.5\textwidth}
%        \centering
%        \includegraphics[scale=0.35]{bayesian_plots/fig11.png}
%        \caption{$c(q) = (q-0.9)^2$}
%		\label{fig:bayesian_plots/fig11}
%    \end{subfigure}
%    \caption{Error exponents as a function of $n$}
%\label{fig:error-exponents-appendix-bayesian}
%\end{figure*}

\subsection{Neyman-Pearson Formulation}
We fix $\varepsilon = 0.1$ and consider the piecewise linear cost function $c(q) = |q-q^*|$ on $Q = [0.7,0.8]$ where $q^* = 0.8$. As suggested by Lemma~\ref{lemma:pureeqforbinary_np}, there exists a pure strategy Nash equilibrium for $\mathcal{G}^{NP}(\varepsilon,2,n)$ for each $n \geq 1$. We first compute the dominant decision rule of the defender by finding the appropriate value of threshold and randomization. Once this is done, we compute the equilibrium by finding the best response of the attacker corresponding to this dominant strategy of the defender (as before, we discretize the set $Q$ into $100$ equally-spaced points). We repeat the experiment for different values of $n$, and for the quadratic cost function $c(q) = 0.001*(q-0.8)^2$. Figure~\ref{fig:np-firstcost} and Figure~\ref{fig:np-secondcost} shows the results for the above two cost functions.

Since the former cost function increases much faster than the latter as we move away from the point $q^*$, we see that the attacker has much more incentive to play a strategy that is away from $q^*$ in the second case compared to the first. This is reflected in the equilibrium strategy of the attacker; from Figures~\ref{fig:np-firstcost}(\subref{fig:np_plots/fig1}) and \ref{fig:np-secondcost}(\subref{fig:np_plots/fig2}), we see that it takes much larger values of $n$ for the equilibrium strategy of the attacker to become equal to $q^*$ in the second case compared to the first. From Figures~\ref{fig:np-firstcost}(\subref{fig:np_plots/fig3}) and \ref{fig:np-secondcost}(\subref{fig:np_plots/fig4}), we see that the error exponents approach the limiting value $D(p||q^*) = 0.223$.

\begin{figure*}[t!]
    \centering
    \begin{subfigure}[t]{0.50\textwidth}
        \centering
        \includegraphics[scale=0.40]{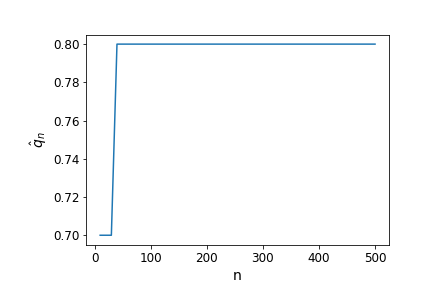}
        \caption{Attacker's strategy as a function of $n$}
        \label{fig:np_plots/fig1}
    \end{subfigure}%
    ~ 
    \begin{subfigure}[t]{0.50\textwidth}
        \centering
        \includegraphics[scale=0.40]{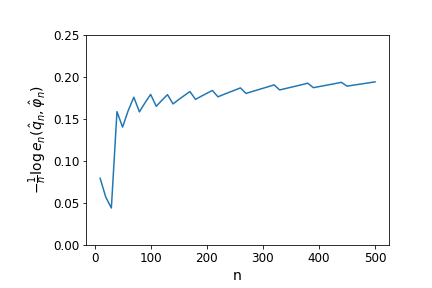}
        \caption{Error exponent as a function of $n$}
        \label{fig:np_plots/fig3}
    \end{subfigure}
    \caption{$c(q) = |q-0.8|$}
\label{fig:np-firstcost}
\end{figure*}
%
%\begin{figure*}[t!]
%    \centering
%    \begin{subfigure}[t]{0.5\textwidth}
%        \centering
%        \includegraphics[scale=0.35]{np_plots/fig1.png}
%        \caption{Attacker's strategy as a function of $n$}
%        \label{fig:np_plots/fig1}
%    \end{subfigure}%
%    ~ 
%    \begin{subfigure}[t]{0.5\textwidth}
%        \centering
%        \includegraphics[scale=0.35]{np_plots/fig3.png}
%        \caption{Error exponent as a function of $n$}
%		\label{fig:np_plots/fig3}
%    \end{subfigure}
%    \caption{$c(q) = |q-0.8|$}
%\label{fig:np-firstcost}
%\end{figure*}

\begin{figure*}[t!]
    \centering
    \begin{subfigure}[t]{0.50\textwidth}
        \centering
        \includegraphics[scale=0.40]{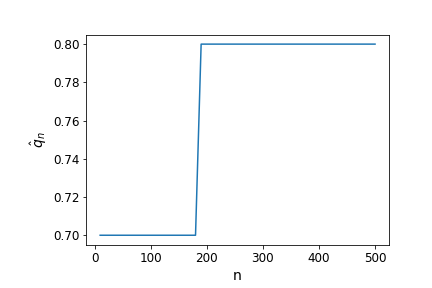}
        \caption{Attacker's strategy as a function of $n$}
        \label{fig:np_plots/fig2}
    \end{subfigure}%
    ~ 
    \begin{subfigure}[t]{0.50\textwidth}
        \centering
        \includegraphics[scale=0.40]{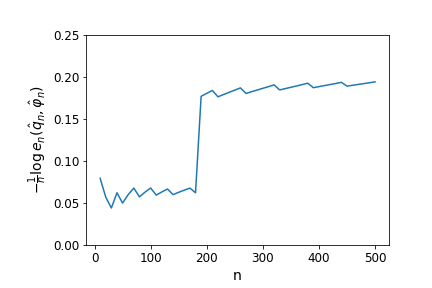}
        \caption{Error exponent as a function of $n$}
        \label{fig:np_plots/fig4}
    \end{subfigure}
    \caption{$c(q) = 0.001*(q-0.8)^2$}
\label{fig:np-secondcost}
\end{figure*}

%
%\begin{figure*}[t!]
%    \centering
%    \begin{subfigure}[t]{0.5\textwidth}
%        \centering
%        \includegraphics[scale=0.35]{np_plots/fig2.png}
%        \caption{Attacker's strategy as a function of $n$}
%        \label{fig:np_plots/fig2}
%    \end{subfigure}%
%    ~ 
%    \begin{subfigure}[t]{0.5\textwidth}
%        \centering
%        \includegraphics[scale=0.35]{np_plots/fig4.png}
%        \caption{Error exponent as a function of $n$}
%		\label{fig:np_plots/fig4}
%    \end{subfigure}
%    \caption{$c(q) = 0.001*(q-0.8)^2$}
%\label{fig:np-secondcost}
%\end{figure*}

}

\end{document}